\documentclass[10pt,aps,pra,showpacs,showkeys,twocolumn]{revtex4-1}
\usepackage{amsmath, amssymb}
\usepackage{amsthm}
\usepackage{graphicx}
\usepackage{cases}
\usepackage{bbold}
\usepackage{booktabs}
\usepackage{dcolumn}
\usepackage{empheq}

\newtheorem{thrm}{Theorem}

\DeclareMathOperator{\re}{Re}
\DeclareMathOperator{\im}{Im}
\DeclareMathOperator{\erf}{erf}
\DeclareMathOperator{\tr}{Tr}

\DeclareMathOperator{\erfi}{erfi}
\DeclareMathOperator{\atanh}{atanh}

\renewcommand{\vec}[1]{\mathbf{#1}}

\begin{document}

\title{Higher-order EPR correlations and inseparability conditions for continuous variables}
\author{E. Shchukin}
\email{evgeny.shchukin@gmail.com}
\author{P. van Loock}
\email{loock@uni-mainz.de}
\affiliation{Johannes-Gutenberg University of Mainz, Institute of Physics, Staudingerweg 7, 55128 Mainz}

\begin{abstract}
We derive two types of sets of higher-order conditions for bipartite entanglement in terms of continuous variables. One 
corresponds to an extension of the well-known Duan inequalities from second to higher moments describing a kind of 
higher-order Einstein-Podolsky-Rosen (EPR) correlations. Only the second type, however, expressed by powers of the mode 
operators leads to tight conditions with a hierarchical structure. We start with a minimization problem for the 
single-partite case and, using the results obtained, establish relevant inequalities for higher-order moments satisfied 
by all bipartite separable states. A certain fourth-order condition cannot be violated by any Gaussian state and we 
present non-Gaussian states whose entanglement is detected by that condition. Violations of all our conditions are 
provided, so they can all be used as entanglement tests.
\end{abstract}



\maketitle

\section{Introduction}

In general, higher-order problems (i.e., of order higher than two) are notoriously difficult to deal with, especially if 
we want to get exact solutions or obtain precise estimations. Here we start with the problem of finding the minimal 
value of the simple quantity $\langle\hat{x}^{2n} + \hat{p}^{2n}\rangle$ for $n \geq 1$, where we use the convention 
$[\hat{x}, \hat{p}] = i$. In other words, we are interested in a kind of uncertainty relation for quantum mechanical 
position $\hat{x}$ and momentum $\hat{p}$ operators, in particular, expressed in terms of moments greater than two. The 
operators $\hat{x}$ and $\hat{p}$ may correspond to the Hermitian quadrature operators of an optical mode. This problem 
was studied in Ref.~\cite{JMP-31-1947}, though the approach there was rather na\"{i}ve (computing eigenvalues by finding 
the roots of the characteristic equation gives very imprecise results even for matrices of moderate size). We improve 
and extend the results obtained there to apply them to derive inequalities for higher-order moments of bipartite 
separable states, which is the actual goal of the present work. 

In the bipartite case, violation of the well-known Duan criteria \cite{PhysRevLett.84.2722} in terms of second-order 
moments, $\langle (\hat{x}_a \pm \hat{x}_b)^2 + (\hat{p}_a \mp \hat{p}_b)^2 \rangle \geq 2$, is sufficient for the 
inseparability of an arbitrary two-mode state and even necessary for that of a two-mode Gaussian state (in a standard 
form). The second-order Duan conditions, originally derived through Cauchy-Schwarz and Heisenberg uncertainty 
inequalities independent of partial transposition, have been shown to be a special case of a hierarchy of bipartite 
inseparability conditions (expressed in terms of arbitrary moments of mode operators) that are based on partial 
transposition \cite{PhysRevLett.95.230502, PhysRevLett.96.050503, PhysRevA.74.032333}. Nonetheless, a nice and unique 
feature of the Duan inequalities  is their resemblance with EPR-type continuous-variable (CV) correlations, where a 
perfect violation of the Duan conditions corresponds to an infinitely squeezed Einstein-Podolsky-Rosen (EPR) state.  

Here we consider the possibility of having similarly intuitive, higher-order inseparability conditions for $\hat{x}$ and 
$\hat{p}$ (a kind of higher-order Duan/EPR criteria for the standard quadrature operators) \cite{PhysRevLett.114.100403, 
PhysRevLett.108.030503}. Note that in Ref.~\cite{PhysRevLett.114.100403} the converse situation of standard EPR 
correlations between higher-order quadratures has been studied. In addition, we will find a set of higher-order 
conditions in terms of the mode operators $\hat{a}$, $\hat{a}^\dagger$, which allow for a more systematic procedure 
(while being less directly accessible in an experiment). We demonstrate that all bipartite Gaussian states, regardless 
if inseparable or not, satisfy our fourth-order separability condition which, nevertheless, can be perfectly violated 
by non-Gaussian states. This condition is thus an example of a truly higher-order condition. We also make a similar 
conjecture about greater-than-fourth-order conditions.

The paper is organized as follows. In Section II we study the higher-order single-partite problem, show how to 
numerically compute the minimal values of $\langle \hat{x}^{2n} + \hat{p}^{2n} \rangle$ for $n>1$ and the corresponding 
wave functions, and give a conjecture about the form of these minimal solutions. In Section III we use the results from 
the preceding section to obtain a higher-order analogue of a well-known bipartite separability condition. In Section IV 
we present another approach to such conditions and give an example of a true higher-order separability condition that 
can be perfectly violated, while no Gaussian state violates it. We give a strict proof only for the fourth-order 
condition and make a conjecture for higher orders. A summary of all the results obtained is given in the conclusion. 
Appendices contain the technical details and computations.

\section{Single-partite case}

We are going to find the minimal value $\lambda^{(2n)}_{\mathrm{min}}>0$ of the quantity $\langle\hat{x}^{2n} + 
\hat{p}^{2n}\rangle$, $n \geq 1$, over all possible physical quantum states $\hat{\varrho}$. In other words, we are 
going to establish the tight inequality of the form
\begin{equation}\label{eq:xp1}
    \langle\hat{x}^{2n} + \hat{p}^{2n}\rangle \geq \lambda^{(2n)}_{\mathrm{min}}.
\end{equation}
Due to linearity of this quantity with respect to $\hat{\varrho}$ (with any $\hat{\varrho}$ being a convex combination 
of pure states $|\psi\rangle$), it is enough to consider pure states only:
\begin{equation}\label{eq:min}
    \min_{\hat{\varrho}}\langle\hat{x}^{2n} + \hat{p}^{2n}\rangle = 
    \min_{|\psi\rangle}\langle\hat{x}^{2n} + \hat{p}^{2n}\rangle.
\end{equation}
To find states that minimize this quantity we must solve the eigenvalue equation
\begin{equation}
    (\hat{x}^{2n} + \hat{p}^{2n})|\psi\rangle = \lambda|\psi\rangle, 
\end{equation}
and find the normalizable solutions(s) corresponding to the minimal eigenvalue $\lambda^{(2n)}_{\mathrm{min}}$ for which 
normalizable solutions exist. In terms of wave functions this equation reads as the following ordinary differential 
equation of $2n$-th order:
\begin{equation}\label{eq:eve}
    x^{2n} \psi(x) + (-1)^n \psi^{(2n)}(x) = \lambda \psi(x),
\end{equation}
with $\psi^{(2n)}(x) \equiv (\partial^{2n}\psi/\partial x^{2n})(x)$. To our best knowledge, solutions of this equation 
are known only for $n=1$. In this case Eq.~\eqref{eq:eve} becomes the Weber equation \cite{*[][{ p. 116, Eq.~(1)}] 
bateman2-1} after rescaling the argument: $\psi^{\prime\prime} + (\lambda - x^2) \psi = 0$. The general solution of this 
equation is given by
\begin{equation}
    \psi(x) = c_1 D_{-\frac{1-\lambda}{2}}(\sqrt{2}x) + c_2 D_{-\frac{1+\lambda}{2}}(i\sqrt{2}x),
\end{equation}
where $D_\nu(x)$ are parabolic cylinder functions \cite{*[][{ p. 117, Eq.~(2)}] bateman2-2} and $c_1, c_2$ are arbitrary 
complex numbers. For $\lambda=1$ this expression simplifies to
\begin{equation}
    \psi(x) = c_1 e^{-x^2/2} + c'_2 e^{-x^2/2} \erfi(x).
\end{equation}
The only normalizable function of this form is $c_1 e^{-x^2/2}$ and after normalization it coincides with the wave 
function of the vacuum state.

Note that if $\psi(x)$ is a solution to Eq.~\eqref{eq:eve} then its complex conjugate $\psi^*(x)$ is also a solution to 
that equation (since $\lambda$ is real) and thus the real part $\re \psi(x) = (1/2)(\psi(x) + \psi^*(x))$ is a real 
solution of Eq.~\eqref{eq:eve}. We now prove a couple of general facts about the solutions of that equation. First, we 
show that any real normalizable solution of Eq.~\eqref{eq:eve} satisfies the property
\begin{equation}\label{eq:x2p}
    \langle\hat{x}^{2n}\rangle = \langle\hat{p}^{2n}\rangle = \frac{\lambda}{2}.
\end{equation}
In fact, if we multiply both sides of Eq.~\eqref{eq:eve} by $x \psi'(x)$ and integrate over $(-\infty, +\infty)$, we get
\begin{equation}\label{eq:xpa}
\begin{split}
    \int x^{2n+1} &\psi(x) \psi'(x) \, dx \\
    &+ (-1)^n \int x \psi'(x) \psi^{(2n)}(x) \, dx = -\frac{\lambda}{2}.
\end{split}
\end{equation}
The first term on the left-hand side is easy to compute,
\begin{equation}
\begin{split}
    \int &x^{2n+1} \psi(x) \psi'(x) \, dx \\
    &= -\frac{2n+1}{2} \int x^{2n} \psi^2(x) \, dx = -\frac{2n+1}{2} \langle\hat{x}^{2n}\rangle.
\end{split}
\end{equation}
The second term requires more efforts, but by induction one can obtain the following expression for it:
\begin{equation}
\begin{split}
    \int x \psi'(x) \psi^{(2n)}(x) \, dx &= (-1)^n \frac{2n-1}{2} \int \psi^{(n) 2}(x) \, dx \\
    &= (-1)^n \frac{2n-1}{2} \langle\hat{p}^{2n}\rangle. \nonumber
\end{split}
\end{equation}
Substituting these expression into Eq.~\eqref{eq:xpa}, we get
\begin{equation}
    (2n+1) \langle\hat{x}^{2n}\rangle - (2n-1) \langle\hat{p}^{2n}\rangle = \lambda.
\end{equation}
On the other hand, if we multiply both sides of Eq.~\eqref{eq:eve} by $\psi(x)$ and integrate, we get 
$\langle\hat{x}^{2n}\rangle + \langle\hat{p}^{2n}\rangle = \lambda$. From these two equations we immediately obtain 
Eq.~\eqref{eq:x2p}.

Next, we prove that the minimal value \eqref{eq:min} is invariant under translations in phase space, i.e., for any real 
numbers $x_0$ and $p_0$ we have
\begin{equation}\label{eq:xpxp}
    \min_{\hat{\varrho}}\langle\hat{x}^{2n} + \hat{p}^{2n}\rangle = 
    \min_{\hat{\varrho}}\langle(\hat{x} - x_0)^{2n} + (\hat{p} - p_0)^{2n}\rangle.
\end{equation}
To prove this equation, we just show that for any state $\hat{\varrho}$ there is another state $\hat{\varrho}'$ such 
that 
\begin{equation}\label{eq:xp-delta}
    \langle(\hat{x} - x_0)^n\rangle = \langle\hat{x}^n\rangle', \quad
    \langle(\hat{p} - p_0)^n\rangle = \langle\hat{p}^n\rangle',
\end{equation}
for all $n \geq 1$. In fact, let us take an arbitrary state $\hat{\varrho}$ and consider the new state $\hat{\varrho}' = 
D(\alpha_0) \hat{\varrho} D^\dagger(\alpha_0)$, where $\alpha_0 = -(x_0 + i p_0)/\sqrt{2}$ and $D(\alpha)$ is the 
displacement operator, which is for any complex number $\alpha$ defined by $D(\alpha) = e^{\alpha \hat{a}^\dagger - 
\alpha^* \hat{a}}$. This operator shifts the position and momentum operators as
\begin{equation}
\begin{split}
    D^\dagger(\alpha) \hat{x} D(\alpha) &= \hat{x} + \sqrt{2} \re\alpha, \\
    D^\dagger(\alpha) \hat{p} D(\alpha) &= \hat{p} + \sqrt{2} \im\alpha.
\end{split}
\end{equation}
Using these relations, we can write
\begin{equation}
\begin{split}
    \langle\hat{x}^n\rangle' &= \tr(\hat{x}^n D(\alpha_0) \hat{\varrho} D^\dagger(\alpha_0)) \\
    &= \tr(D^\dagger(\alpha_0) \hat{x}^n D(\alpha_0) \hat{\varrho}) = \langle(\hat{x} - x_0)^n\rangle,
\end{split}
\end{equation}
so we get the first equality of Eq.~\eqref{eq:xp-delta}. The second one is obtained in the same way. We see that any 
number of the form $\langle(\hat{x} - x_0)^{2n} + (\hat{p} - p_0)^{2n}\rangle$ is also of the form $\langle\hat{x}^{2n} 
+ \hat{p}^{2n}\rangle$. Using the operator $D^\dagger(\alpha_0) = D(-\alpha_0)$ instead of $D(\alpha_0)$, we can also 
conclude that the inverse statement is true --- any number of the form $\langle\hat{x}^{2n} + \hat{p}^{2n}\rangle$ is 
also of the form $\langle(\hat{x} - x_0)^{2n} + (\hat{p} - p_0)^{2n}\rangle$, and thus Eq.~\eqref{eq:xpxp} holds.

We now present a method to obtain an analytical lower bound on the quantity $\langle\hat{x}^{2n}+\hat{p}^{2n}\rangle$ 
for $n=2$. We can expand $\langle\hat{x}^4+\hat{p}^4\rangle$ in terms of the creation and annihilation operators as 
follows:
\begin{equation}\label{eq:x4p4-1}
    \langle \hat{x}^4 + \hat{p}^4 \rangle = \frac{3}{2} + \frac{1}{2} \langle \hat{a}^4 + 
    \hat{a}^{\dagger 4} \rangle + 3 \langle \hat{a}^{\dagger 2} \hat{a}^2 \rangle + 6 \langle \hat{a}^\dagger 
    \hat{a} \rangle.
\end{equation}
The fourth powers can be estimated as
\begin{equation}\label{eq:a4}
\begin{split}
    |\langle\hat{a}^4\rangle|^2 &\leq \langle\hat{a}^{\dagger 2}\hat{a}^2\rangle 
    \langle\hat{a}^2\hat{a}^{\dagger 2}\rangle \\
    &= \langle\hat{a}^{\dagger 2}\hat{a}^2\rangle 
    (\langle\hat{a}^{\dagger 2}\hat{a}^2\rangle + 4\langle\hat{a}^\dagger\hat{a}\rangle + 2),
\end{split}
\end{equation}
and thus $\langle\hat{x}^4+\hat{p}^4\rangle$ satisfies the inequality
\begin{equation}
    \langle \hat{x}^4 + \hat{p}^4 \rangle \geq \frac{3}{2} + 3A + 6B - \sqrt{A(A+4B+2)},
\end{equation}
where $A = \langle\hat{a}^{\dagger 2}\hat{a}^2\rangle$ and $B = \langle\hat{a}^\dagger\hat{a}\rangle$. We demonstrate 
that
\begin{equation}\label{eq:AB2}
    3A + 6B - \sqrt{A(A+4B+2)} \geq -(3-2\sqrt{2}) \equiv -\delta
\end{equation}
for all $A, B \geq 0$. In fact, this inequality is equivalent to the following one: $3A + 6B + \delta \geq 
\sqrt{A(A+4B+2)}$. Since all the parameters are nonnegative, we can take the square of both sides of this inequality. 
Taking into account that $3\delta-1 = -2\sqrt{2}\delta$, we arrive at an equivalent statement,
\begin{equation}\label{eq:ABdelta}
    (2\sqrt{2}A - \delta)^2 + 36 B^2 + 32 AB + 12 B\delta \geq 0,
\end{equation}
which is obviously valid. The estimation given by Eq.~\eqref{eq:AB2} is the best possible under the restriction $A,B 
\geq 0$: For $A = \delta/2\sqrt{2}, B = 0$ the inequality in Eq.~\eqref{eq:ABdelta} is tight. We have just established 
the following result:
\begin{equation}\label{eq:xp4}
    \langle \hat{x}^4 + \hat{p}^4 \rangle \geq \frac{3}{2} - \delta \approx 1.32843.
\end{equation}
It is possible to extend this approach and apply it to obtain estimations for higher-order combinations of moments, but 
these inequalities are not tight for any $n \geq 2$; more precise results can be obtained by numerical optimization.

For $n>1$ the only way to work with Eq.~\eqref{eq:eve} is numerics. To solve this equation numerically we need a set of 
initial conditions. Many initial conditions lead to nonnormalizable solutions. Since the general solution is unknown, we 
need some way to determine what initial conditions to set to obtain a normalizable wave function.

One way to do this is to minimize $\langle\hat{x}^{2n} + \hat{p}^{2n}\rangle$ as a quadratic form of the coefficients in 
the Fock basis expansion. To write $\langle\hat{x}^{2n} + \hat{p}^{2n}\rangle$ as a quadratic form we need to express 
the powers of the position and momentum operators in terms of the creation and annihilation operators. According to 
Ref.~\cite{*[][{ p. 284, Eq.~(A.14)}] job-1-264}, we have
\begin{equation}\label{eq:xp-n}
\begin{split}
    (\hat{a} + \hat{a}^\dagger)^{2n} &= \sum^{2n}_{k=0} \frac{(2n)!}{k!} 
    \sum^{\lfloor\frac{2n-k}{2}\rfloor}_{l=0} \frac{\hat{a}^{\dagger 2n-k-2l} \hat{a}^k}{2^l l! (2n-k-2l)!}, \\
    (\hat{a} - \hat{a}^\dagger)^{2n} &= \sum^{2n}_{k=0} \frac{(2n)!}{k!} 
    \sum^{\lfloor\frac{2n-k}{2}\rfloor}_{l=0} \frac{(-1)^{k+l}\hat{a}^{\dagger 2n-k-2l} \hat{a}^k}{2^l l! (2n-k-2l)!}.
\end{split}
\end{equation}
From these relations we derive the following results:
\begin{equation}\label{eq:xp246}
\begin{split}
    \hat{x}^2 + \hat{p}^2 &= 1 + 2 \hat{a}^\dagger \hat{a}, \\
    \hat{x}^4 + \hat{p}^4 &= \frac{3}{2} + \Bigl[3\hat{a}^{\dagger 2}\hat{a}^2 + 6\hat{a}^\dagger 
    \hat{a}\Bigr] + \Bigl[\frac{1}{2}(\hat{a}^4 + \hat{a}^{\dagger 4})\Bigr], \\
    \hat{x}^6 + \hat{p}^6 &= \frac{15}{4} + \Bigl[5\hat{a}^{\dagger 3}\hat{a}^3 + 
    \frac{45}{2}\hat{a}^{\dagger 2}\hat{a}^2 + \frac{45}{2}\hat{a}^\dagger \hat{a}\Bigr] \\
    &+ \Bigl[\frac{3}{2}(\hat{a}^\dagger \hat{a}^5 + \hat{a}^{\dagger 5}\hat{a}) + 
    \frac{15}{4}(\hat{a}^4 + \hat{a}^{\dagger 4})\Bigr]. \\
\end{split}
\end{equation}
This list can be continued, but the expressions will become more and more complicated. We see that the right-hand sides 
of these relations contain terms of the form $\hat{a}^{\dagger p} \hat{a}^q$ with $p-q$ being a multiple of 4. It is 
easy to demonstrate that this is true in general, for all $n \geq 1$. In fact, each term on the right-hand side of 
Eq.~\eqref{eq:xp-n} has the form $\hat{a}^{\dagger 2n-k-2l} \hat{a}^k$, so the difference of powers is $2(n-k-l)$. On 
the other hand, the coefficient in front of this term in $\hat{x}^{2n} + \hat{p}^{2n}$ is proportional to $1 + 
(-1)^{n+k+l}$ and is nonzero only if $n + k + l$ is even. In such a case the difference $n - k - l$ is also even and 
thus the difference of powers $2(n-k-l)$ is a multiple of 4. So, independently of $n$, only terms $\hat{a}^{\dagger p} 
\hat{a}^q$ with $p-q$ being a multiple of 4 are present in $\hat{x}^{2n} + \hat{p}^{2n}$. 

The constants on the right-hand side of Eq.~\eqref{eq:xp246} correspond to the values of the quantity on the left-hand 
side for the vacuum state. Thus, note that only for $n=1$ is the vacuum state a minimum uncertainty state, whereas for 
$n>1$ other, especially non-Gaussian states have smaller uncertainties compared to the vacuum. The uncertainty value for 
the vacuum is easy to compute in general and it is given by the following expression: 
\begin{equation}\label{eq:vac}
    \langle\hat{x}^{2n} + \hat{p}^{2n}\rangle_0 = \frac{(2n)!}{2^{2n-1}n!}.
\end{equation}
Below, by the matrix $M_{2n}$ of $\langle\hat{x}^{2n} + \hat{p}^{2n}\rangle$ we mean the matrix of this quantity 
expressed in the Fock basis. As we will see, for $n>1$ this quantity can go below the value $\langle\hat{x}^{2n} + 
\hat{p}^{2n}\rangle_0$.

The simplest quantity of this form, $\langle\hat{x}^2 + \hat{p}^2\rangle$, is already diagonal in the Fock basis, the 
minimal eigenvalue being 1. The matrices of $\langle\hat{x}^{2n} + \hat{p}^{2n}\rangle$ with $n > 1$ have a simple 
structure --- if $n = 2m$ or $n = 2m+1$ then the matrix of $\langle\hat{x}^{2n} + \hat{p}^{2n}\rangle$ has $2m+1$ 
diagonals, where $m$ of them are below the main diagonal and $m$ of them are above. The distance between the adjacent 
diagonals is 4. For example, the matrix $M_4 = (a_{i,j})$ of $\langle\hat{x}^4 + \hat{p}^4\rangle$ has three diagonals 
and explicitly it reads as follows:
\begin{equation}\label{eq:M4}
    M_4 = \frac{3}{2} \mathbb{1} + 
    \begin{pmatrix}
        0 & 0 & 0 & 0 & \sqrt{6} & 0 & 0 & \ldots \\     
        0 & 6 & 0 & 0 & 0 & \sqrt{30} & 0 & \ldots  \\
        0 & 0 & 18 & 0 & 0 & 0 & \sqrt{90} & \ldots  \\
        0 & 0 & 0 & 36 & 0 & 0 & 0 & \ldots  \\
        \sqrt{6} & 0 & 0 & 0 & 60 & 0 & 0 & \ldots  \\
        0 & \sqrt{30} & 0 & 0 & 0 & 90 & 0 & \ldots  \\
        0 & 0 & \sqrt{90} & 0 & 0 & 0 & 126 & \ldots  \\
        \hdotsfor{8}
    \end{pmatrix},
\end{equation}
where $\mathbb{1}$ is the identity matrix and the nonzero elements are given by the following equations:
\begin{displaymath}
\begin{split}
    a_{k,k} &= \frac{3}{2} + 3k(k+1), \\
    a_{k,k+4} &= \frac{1}{2}\sqrt{(k+1)(k+2)(k+3)(k+4)}.
\end{split}
\end{displaymath}
The matrix $M_6$ has the same structure, but its nonzero elements become
\begin{displaymath}
\begin{split}
    a_{k,k} &= \frac{15}{4} + 5k\left(k^2 + \frac{3}{2}k + 2\right), \\
    a_{k,k+4} &= \frac{3}{2}\left(k + \frac{5}{2}\right)\sqrt{(k+1)(k+2)(k+3)(k+4)}.
\end{split}  
\end{displaymath}
The matrices $M_8$ and $M_{10}$ have five diagonals, $M_{12}$ and $M_{14}$ have seven, and so on. The elements of these 
matrices can be obtained with the help of Eq.~\eqref{eq:xp-n}. It can be now easily seen that for $n>1$ the quantity 
$\langle\hat{x}^{2n} + \hat{p}^{2n}\rangle$ goes below the value for the vacuum state. In fact, for $n>1$ there are at 
least two additional diagonals which are 4 positions below and above the main diagonal. The first element of the main 
diagonal (with indices $00$) is zero, and taking the state of the form $|\psi\rangle = N(|0\rangle + c|4\rangle$) we 
immediately see that for small negative values of $c$ the value of the quantity $\langle\hat{x}^{2n} + 
\hat{p}^{2n}\rangle$ is smaller than the value given by Eq.~\eqref{eq:vac}.

\begin{figure}[ht]
    \includegraphics{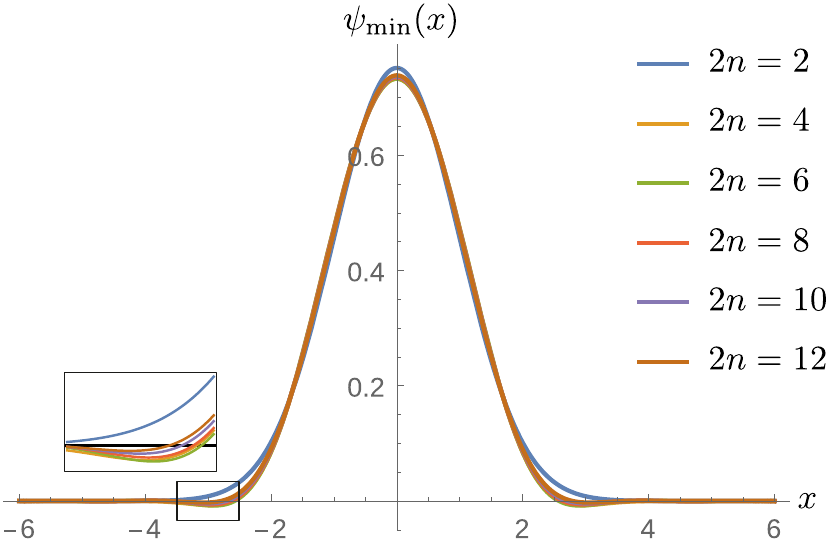}
    \caption{(Color online). The wave functions of the states minimizing the quantity $\langle\hat{x}^{2n} + 
    \hat{p}^{2n}\rangle$ for a few values of $n$.}\label{fig:xp-n}
\end{figure}

\begin{table}[ht]
\begin{tabular}{l|@{\extracolsep{3mm}}llllll}
\toprule[0.6pt]
 $2n$ & \hfil 2 & \hfil 4 & \hfil 6 & \hfil 8 & \hfil 10 & \hfil 12 \\
\midrule[0.4pt]\\[-2mm]
$\lambda^{(2n)}_{\mathrm{min}}$ & 1 & 1.3967 & 2.9530 & 8.2891 & 28.9741 &  
121.2168 \\[1mm]
\bottomrule[0.6pt]
\end{tabular}
\caption{Minimal eigenvalue for the first few values of $n$.}\label{tbl:lambda}
\end{table}

The minimal eigenvalue $\lambda^{(2n)}_{\mathrm{min}}$ of these matrices and the corresponding eigenvectors $(c_0, c_1, 
\ldots)$ (coefficients in the Fock basis) are easy to compute numerically. These values for small $n$ are given in 
Table~\ref{tbl:lambda}. Full details of the numerics to perform this computation are given in Appendix A. The 
eigenvalues are shown in Fig.~\ref{fig:psi1min-pic} together with the quantity $\langle\hat{x}^{2n} + 
\hat{p}^{2n}\rangle_0$ for the same $n$. The figure is in logarithmic scale, where a linear dependence would mean an 
exponential growth. According to this figure the minimal eigenvalues grow faster than any linear function, so that the 
minimal value of $\langle\hat{x}^{2n} + \hat{p}^{2n}\rangle$ increases faster than exponentially. 

The wave function 
\begin{equation}\label{eq:psi-c}
    \psi_{\mathrm{min}}(x) = \sum^{+\infty}_{k=0} c_k \psi_k(x)
\end{equation}
of the minimal state is shown in Fig.~\ref{fig:xp-n} for $n = 1, \ldots, 6$, where $\psi_k(x)$ are the wave functions of 
the Fock states. It can be seen that the functions for $n > 1$ are nearly indistinguishable and rather close to the 
vacuum wave function (i.e., to the solution for $n=1$). The main difference between the wave function of the vacuum 
state and the minimal wave functions for $n>1$ is that the latter take negative values.

\begin{figure}[ht]
    \includegraphics{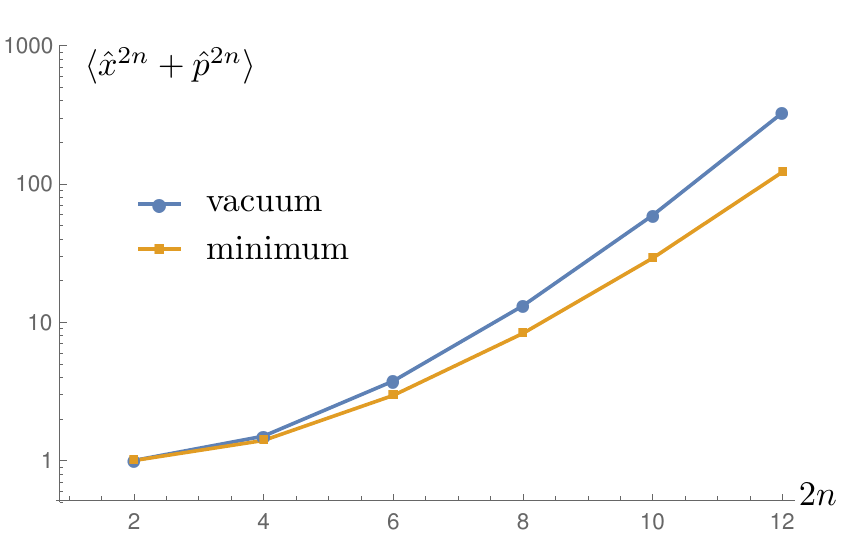}
    \caption{(Color online). Superexponential growth of the minimal eigenvalues $\lambda^{(n)}_{\mathrm{min}}$. Note 
    that the $y$-axis is in $\log$-scale.}\label{fig:psi1min-pic}
\end{figure}

It happens that the wave functions of the minimizing states can be accurately (with relative error $\approx$ 1\%) 
approximated by the following expression:
\begin{equation}\label{eq:BJ0}
    \psi_{a,b}(x) = c J_0(a x^2) e^{-b x^2},
\end{equation}
where $a$ and $b$ are appropriately chosen positive parameters and $c$ is determined from the normalization of 
$\psi_{a,b}(x)$. For $n=1$ this expression with $a=0$ and $b=1/2$ exactly reproduces the wave function of the vacuum 
state. The normalization is explicitly given by
\begin{equation}
    c = \frac{\pi^{3/4}}{2} 
    \frac{\sqrt[4]{\sqrt{a^2+b^2}+b}}{K\left(\sqrt{\frac{a-\sqrt{2b}\sqrt{\sqrt{a^2+b^2}-b}}{2a}}\right)}, 
\end{equation}
where $K(k)$ is the complete elliptic integral of the first kind defined by the following expression:
\begin{equation}
    K(k) = \int^{\pi/2}_0 \frac{d\theta}{\sqrt{1-k^2 \sin^2\theta}}.
\end{equation}
It is easy to verify the relations
\begin{equation}
    0 < \frac{a-\sqrt{2b}\sqrt{\sqrt{a^2+b^2}-b}}{2a} < \frac{1}{2},
\end{equation}
which are valid for all positive $a$ and $b$, and the relation
\begin{equation}
    a-\sqrt{2b}\sqrt{\sqrt{a^2+b^2}-b} = \frac{a^3}{8b^2} + O(a^4),
\end{equation}
valid for a fixed positive $b$, from which we derive that for $a = 0$ and $b = 1/2$ we get $c = \pi^{-1/4}$, as it must 
be for the vacuum state. For $n>1$ the expression \eqref{eq:BJ0} gives only an approximation to the exact minimizing 
state. In Table~\ref{tbl:1} we present the parameters $a$ and $b$ for the first few values of $n$.

\begin{table}[ht]
\begin{tabular}{r|@{\extracolsep{3mm}}lll}
\toprule[0.6pt]
 $2n$ & \hfil $a$ & \hfil $b$ & \hfil $c$ \\
\midrule[0.4pt]\\[-2mm]
2 & 0 & 0.5 & 0.751126 \\[1mm]
4 & 0.345424 & 0.402533 & 0.731575 \\[1mm]
6 & 0.350766 & 0.399127 & 0.730834 \\[1mm]
8 & 0.334137 & 0.409370 & 0.733031 \\[1mm]
10 & 0.314942 & 0.420320 & 0.735346 \\[1mm]
12 & 0.297065 & 0.429728 & 0.737304 \\[1mm]
\bottomrule[0.6pt]
\end{tabular}
\caption{Parameters $a$ and $b$.}\label{tbl:1}
\end{table}

To sum up, we have established a kind of higher-order uncertainty relation, which, however, is less general than the 
positivity of the density matrix. It generalizes known uncertainty relations by incorporating moments greater than two. 
Finally, at the end of this section, let us prove the following relation for the minimal eigenvalues:
\begin{equation}\label{eq:ll2}
    \lambda^{(4n)}_{\mathrm{min}} > \frac{1}{2} \left(\lambda^{(2n)}_{\mathrm{min}}\right)^2.
\end{equation}
In fact, due to the inequality $\langle\hat{A}^2\rangle \geq \langle\hat{A}\rangle^2$ we have
\begin{equation}
\begin{split}
    \lambda^{(4n)}_{\mathrm{min}} &= \langle\hat{x}^{4n}+\hat{p}^{4n}\rangle_{\mathrm{min}} \geq
    \langle\hat{x}^{2n}\rangle^2_{\mathrm{min}} + \langle\hat{p}^{2n}\rangle^2_{\mathrm{min}} \\
    &\geq \frac{1}{2}\left(\langle\hat{x}^{2n}\rangle_{\mathrm{min}} +
    \langle\hat{p}^{2n}\rangle_{\mathrm{min}}\right)^2 > \frac{1}{2}\left(\lambda^{(2n)}_{\mathrm{min}}\right)^2.
\end{split}
\end{equation}
Note that the subindex $\mathrm{min}$ here refers to the state corresponding to the minimal value 
$\lambda^{(4n)}_{\mathrm{min}}$, and it is not the state that minimizes the quantity 
$\langle\hat{x}^{2n}+\hat{p}^{2n}\rangle$, so, by definition of $\lambda^{(2n)}_{\mathrm{min}}$, we get the last step of 
these relations. Table~\ref{tbl:lambda} shows that in fact the minimal values $\lambda^{(2n)}_{\mathrm{min}}$ grow 
faster than is guaranteed by the inequality \eqref{eq:ll2}. This inequality will be useful to demonstrate the 
higher-order separability conditions we construct in the next section are not trivial consequence of lower-order 
conditions.

\section{Bipartite case}

We start with a well-known result \cite{PhysRevLett.84.2722}, which is valid for all bipartite separable states,
\begin{equation}\label{eq:xp-2}
    \langle (\hat{x}_a \pm \hat{x}_b)^2 + (\hat{p}_a \mp \hat{p}_b)^2 \rangle \geq 2.
\end{equation}
There are at least two possible ways to extend this inequality to higher orders. We develop them in the subsections that 
follow. The first one is easy to implement experimentally and easy to violate, but it is not so easy to obtain the 
optimal result. The main disadvantage of this approach is that it is not a ``true" hierarchy of conditions for 
higher-order moments, since all of these conditions can be violated by Gaussian states. Nevertheless, the higher-order 
inequalities we derive are stronger than those based on only second-order moments. The other approach leads to tight 
conditions, but these conditions are more difficult to implement. On the other hand, these conditions may be referred to 
as truly higher-order as they cannot be violated by Gaussian states. We have confirmed this by numerical simulation 
while we were able to strictly prove this only in the simplest case of the fourth-order condition. Note that an example 
of the converse situation, namely non-Gaussian states whose entanglement cannot be detected via the (standard) 
second-order conditions but only via fourth-order conditions, is given in Ref.~\cite{PhysRevLett.95.230502}.

\subsection{Approach one}

The most obvious way to extend inequality \eqref{eq:xp-2} is to replace second powers by higher numbers and try to 
establish an inequality of the form
\begin{equation}\label{eq:xp2n}
    \langle (\hat{x}_a \pm \hat{x}_b)^{2n} + (\hat{p}_a \mp \hat{p}_b)^{2n} \rangle \geq ?,
\end{equation}
with a positive bound on the right-hand side. 

Unfortunately, it is rather difficult to find the tight lower bound of the left-hand side of this inequality over all 
separable states. A suboptimal result can be obtained by noting that 
\begin{equation}
\begin{split}
    \langle (\hat{x}_a \pm \hat{x}_b)^{2n} &+ (\hat{p}_a \mp \hat{p}_b)^{2n} \rangle^{\mathrm{PT}} \\
    &= \langle (\hat{x}_a \pm \hat{x}_b)^{2n} + (\hat{p}_a \pm \hat{p}_b)^{2n} \rangle,
\end{split}
\end{equation}
where superscript $\mathrm{PT}$ means the partially transposed state, and finding the minimal value of the quantities
\begin{equation}\label{eq:xpn}
    \langle (\hat{x}_a \pm \hat{x}_b)^{2n} + (\hat{p}_a \pm \hat{p}_b)^{2n} \rangle
\end{equation}
over the set of all bipartite quantum states (representing a physicality bound for all PPT states with regards to the 
original combination \eqref{eq:xp2n} and hence yielding a necessary condition for all separable states). The difference 
in the problem of minimizing over all quantum states and the problem of minimizing over separable states is that the 
former reduces to minimizing a quadratic form, which is straightforward to do numerically, while the latter reads as a 
minimization of a biquadratic form, for which no numerical technique exists. The former value can be obtained with the 
same approach that we used in the previous section for single-partite quantities. As we will see below, this bipartite 
minimal value has a simple relation to the single-partite one, as in the previous section.

To find states that minimize the quantity \eqref{eq:xpn} for some $n$ we have to solve the eigenvalue problem
\begin{equation}\label{eq:wf-bip}
\begin{split}
    (x \pm y)^{2n}\psi(x,y) &+ (-1)^n \left(\frac{\partial}{\partial x} \pm 
    \frac{\partial}{\partial y}\right)^{2n}\psi(x,y) \\
    &= \Lambda \psi(x,y),
\end{split}
\end{equation}
and find the minimal eigenvalue $\Lambda$. These equations look more difficult than Eq.~\eqref{eq:eve}, but they can be 
easily reduced to that equation. In fact, let us introduce the function $\tilde{\psi}(u, v)$ via
\begin{equation}\label{eq:psi-tilde}
    \tilde{\psi}(u, v) = \psi\left(\frac{u+v}{\sqrt{2}}, \pm \frac{u-v}{\sqrt{2}}\right),
\end{equation}
where the choice of the sign corresponds to the sign in Eq.~\eqref{eq:xp2n}. This function is normalized and thus it is 
also a wave function. The relation 
\eqref{eq:psi-tilde}
is invertible
\begin{equation}
    \psi(x, y) = \tilde{\psi}\left(\frac{x \pm y}{\sqrt{2}}, \frac{x \mp y}{\sqrt{2}}\right),
\end{equation}
from which we obtain the following equality:
\begin{equation}
    \left(\frac{\partial}{\partial x} \pm \frac{\partial}{\partial y}\right)^{2n} \psi(x, y) = 
    2^n \frac{\partial^{2n} \tilde{\psi}}{\partial u^{2n}}(u, v),
\end{equation}
where $u = (x \pm y)/\sqrt{2}$ and $v = (x \mp y)/\sqrt{2}$. Substituting this into Eq.~\eqref{eq:wf-bip}, we get an 
equation for $\tilde{\psi}$
\begin{equation}\label{eq:eve2}
    u^{2n}  \tilde{\psi} + (-1)^n \frac{\partial^{2n} \tilde{\psi}}{\partial u^{2n}} = 
    \frac{\Lambda}{2^n} \tilde{\psi},
\end{equation}
which looks very similar to Eq.~\eqref{eq:eve}. Since the minimal solution of that equation is unique, minimal solutions 
of Eq.~\eqref{eq:eve2} are given by
\begin{equation}
    \tilde{\Psi}_{\mathrm{min}}(u, v) = \psi_{\mathrm{min}}(u) \varphi(v),
\end{equation}
where $\psi_{\mathrm{min}}(u)$ is the minimal solution of Eq.~\eqref{eq:eve}, given by Eq.~\eqref{eq:psi-c}, and 
$\varphi(v)$ is an arbitrary normalized function. The minimal solutions of Eq.~\eqref{eq:wf-bip} then become
\begin{equation}\label{eq:Psi-bip}
    \Psi_{\mathrm{min}}(x, y) = \psi_{\mathrm{min}}\left(\frac{x \pm y}{\sqrt{2}}\right) 
\varphi\left(\frac{x \mp y}{\sqrt{2}}\right).
\end{equation}
The bipartite minimal eigenvalue in both cases is just the appropriately scaled single-partite minimal eigenvalue: 
\begin{equation}\label{eq:Ll}
    \Lambda^{(2n)}_{\mathrm{min}} = 2^n \lambda^{(2n)}_{\mathrm{min}}. 
\end{equation}
The bipartite minimal eigenvalues can be numerically computed independently with the same approach as for the 
single-partite case, by minimizing the quantities defined by Eq.~\eqref{eq:xpn} as quadratic forms with respect to the 
bipartite Fock basis. The numerical results agree with the analytical relation \eqref{eq:Ll}. 

\begin{table}[ht]
\begin{tabular}{l|@{\extracolsep{3mm}}llllll}
\toprule[0.6pt]
 $2n$ & \hfil 2 & \hfil 4 & \hfil 6 & \hfil 8 & \hfil 10 & \hfil 12 \\
\midrule[0.4pt]\\[-2mm]
$\Lambda^{(2n)}_{\mathrm{min}}$ & 2 & 5.5868 & 23.624 & 132.626 & 927.171 &  
7757.88 \\[1mm]
\bottomrule[0.6pt]
\end{tabular}
\caption{Minimal eigenvalue for the first few values of $n$.}\label{tbl:lambda2}
\end{table}

We have established the following inequalities for all bipartite separable states:
\begin{equation}\label{eq:xpn-sep}
    \langle (\hat{x}_a \pm \hat{x}_b)^{2n} + (\hat{p}_a \mp \hat{p}_b)^{2n} \rangle \geq   
    \Lambda^{(2n)}_{\mathrm{min}}.
\end{equation}
From Eqs.\eqref{eq:Ll} and \eqref{eq:xp4} we have $\Lambda^{(4)}_{\mathrm{min}} = 4 \lambda^{(4)}_{\mathrm{min}} \geq 
5.3137$, so that we have analytically established the inequality
\begin{equation}
    \langle (\hat{x}_a \pm \hat{x}_b)^4 + (\hat{p}_a \mp \hat{p}_b)^4 \rangle \geq 2(4\sqrt{2} - 3) \approx 5.3137
\end{equation}
for all bipartite separable states. This inequality is not tight, but this result is obtained analytically. A better 
estimation (obtained numerically) can be taken from Table~\ref{tbl:lambda} and it reads as
\begin{equation}\label{eq:xp4-1}
    \langle (\hat{x}_a \pm \hat{x}_b)^4 + (\hat{p}_a \mp \hat{p}_b)^4 \rangle \gtrsim 5.5868.
\end{equation}
We see that the analytical result is rather close to the more precise lower bound found numerically. But even this lower 
bound, as well as Eq.~\eqref{eq:xpn-sep} in general, is unlikely to be tight, but nevertheless these inequalities 
represent some nontrivial tests for higher-order moments. From Table~\ref{tbl:lambda} of the single-partite minimal 
values and  relation \eqref{eq:Ll}, we derive  Table~\ref{tbl:lambda2} of the bipartite (not necessarily tight) lower 
bounds. Note that the bounds in Table~\ref{tbl:lambda2} are obtained with the help of partial transposition. The true 
minimal values with regards to the left-hand side of Eq.~\eqref{eq:xp2n} may always be larger, but will never be as 
large as the corresponding value for the vacuum state.

Here we should make an important observation: If we have a lower bound of the form \eqref{eq:xpn-sep}, then we can 
immediately obtain the following lower bound for $\langle (\hat{x}_a \pm \hat{x}_b)^{4n} + (\hat{p}_a \mp 
\hat{p}_b)^{4n} \rangle$ in the same way as we derived inequality \eqref{eq:ll2} for the single-partite case:
\begin{equation}
\begin{split}
    &\langle (\hat{x}_a \pm \hat{x}_b)^{4n} + (\hat{p}_a \mp \hat{p}_b)^{4n} \rangle \\
    & \geq \langle (\hat{x}_a \pm \hat{x}_b)^{2n}\rangle^2 + \langle(\hat{p}_a \mp \hat{p}_b)^{2n} \rangle^2 \\
    & \geq \frac{1}{2}\langle (\hat{x}_a \pm \hat{x}_b)^{2n} + (\hat{p}_a \mp \hat{p}_b)^{2n} \rangle^2 \geq
    \frac{1}{2} \bigl(\Lambda^{(2n)}_{\mathrm{min}}\bigr)^2.
\end{split}
\end{equation}
This estimation can be most easily made provided that we know only the relation in Eq.~\eqref{eq:xpn-sep} without any 
additional assumptions. The difference between the inequalities \eqref{eq:xp1} and \eqref{eq:xpn-sep} is that the former 
is a general property of physical systems (provided that quantum mechanics gives an adequate description of the physical 
world), while the latter is a property of bipartite separable states, i.e., it is a condition that can be tested against 
all quantum states. Those states that fail this test are thus verified to be entangled. We conclude that for the 
inequalities \eqref{eq:xpn-sep} to form a nontrivial hierarchy of conditions, the minimal values 
$\Lambda^{(2n)}_{\mathrm{min}}$ must satisfy the strict inequalities
\begin{equation}\label{eq:LL2}
    \Lambda^{(4n)}_{\mathrm{min}} > \frac{1}{2} \left(\Lambda^{(2n)}_{\mathrm{min}}\right)^2.
\end{equation}
This is the main requirement for the minimal eigenvalues so that no condition of the form \eqref{eq:xpn-sep} is a 
trivial consequence of another one. The inequality \eqref{eq:ll2} combined with the relation \eqref{eq:Ll} gives us 
exactly the desired result \eqref{eq:LL2}. From Table~\ref{tbl:lambda2} we see that the numbers 
$\Lambda^{(2n)}_{\mathrm{min}}$ we obtained numerically grow much faster than given by the main requirement, and thus 
the inequalities \eqref{eq:xpn-sep} form a hierarchy of separability conditions where indeed the power of each condition 
increases with its order. Also note that we do not need any new ``hardware'' to perform all the tests given by the 
hierarchy \eqref{eq:xpn-sep}; the same experimental setup developed for testing the simplest inequality \eqref{eq:xp-2} 
can be used to check all the inequalities in the hierarchy. This is one of the biggest advantages of this hierarchy.

In Appendix B we show that the inequality \eqref{eq:xp4-1} can be strengthened by using higher-order uncertainties 
instead of partial transposition
\begin{equation}
    \langle (\hat{x}_a \pm \hat{x}_b)^4 + (\hat{p}_a \mp \hat{p}_b)^4 \rangle \gtrsim 5.7934.
\end{equation}
On the other hand, if we take the factorizable state of the form $|\psi\rangle_a |0\rangle_b$ then
\begin{equation}
    \langle (\hat{x}_a + \hat{x}_b)^4 + (\hat{p}_a - \hat{p}_b)^4 \rangle = 
    6 + \frac{1}{2}\langle \hat{a}^4 + \hat{a}^{\dagger 4} + 6\hat{a}^{\dagger 
    2}\hat{a}^2 + 24 \hat{a}^\dagger \hat{a} \rangle. \nonumber
\end{equation}
The last term can be optimized and its minimal value is equal to $\approx -0.0728$. In fact, even for the state 
$|\psi\rangle_a = N(|0\rangle_a + c|4\rangle)$ we can get the value of $-0.0714$ for a small negative value of $c$. So, 
the true minimal value of $\langle (\hat{x}_a + \hat{x}_b)^4 + (\hat{p}_a - \hat{p}_b)^4 \rangle$ is in the narrow 
interval between $5.7934$ and $5.9272$. Even though we do not obtain the exact solution, from practical point of view 
one can say that it is the end of the story of the quantity of the fourth order. Unfortunately, the method used there 
cannot be easily applied to higher-order moments in a systematic way as we have done with PPT approach.

We now show that the inequalities \eqref{eq:xpn-sep} can be perfectly violated by the bipartite two-mode squeezed vacuum 
state.
\begin{thrm}
For the bipartite squeezed vacuum state defined by
\begin{equation}\label{eq:sq-vac}
    |0\rangle_\lambda = \sqrt{1-\lambda^2} \sum^{+\infty}_{n=0} \lambda^n |n,n\rangle,
\end{equation}
the following relations hold for any $n \geq 1$: 
\begin{equation}\label{eq:xp-sq}
    \langle (\hat{x}_a \pm \hat{x}_b)^{2n} \rangle = \langle (\hat{p}_a \mp \hat{p}_b)^{2n} \rangle 
    = \frac{(2n)!}{2^n n!} \left(\frac{1 \pm \lambda}{1 \mp \lambda}\right)^n.
\end{equation}
The squeezed vacuum state thus perfectly violates the inequalities \eqref{eq:xpn-sep} when $\lambda \to \mp 1$, 
respectively. Indeed, the two-mode squeezed vacuum state becomes a simultaneous zero-eigenstate of $\hat{x}_a \pm 
\hat{x}_b$ and $\hat{p}_a \mp \hat{p}_b$ in the limit $\lambda \to \mp 1$.
\end{thrm}
\begin{proof}
The squeezed vacuum state can be more compactly written in the following way:
\begin{equation}\label{eq:sq-st}
    |0\rangle_\lambda = \sqrt{1-\lambda^2} e^{\lambda \hat{a}^\dagger \hat{b}^\dagger} |0,0\rangle.
\end{equation}
The desired quantities are easy to compute with the help of the generating functions
\begin{equation}
\begin{split}
    \mathcal{X}_\pm(t) &= \langle e^{t(\hat{x}_a \pm \hat{x}_b)} \rangle = e^{t^2/2} 
\langle e^{u^*\hat{a}^\dagger} 
    e^{v^*\hat{b}^\dagger} e^{u\hat{a}} e^{v\hat{b}} \rangle, \\
    \mathcal{P}_\pm(t) &= \langle e^{t(\hat{p}_a \pm \hat{p}_b)} \rangle = e^{t^2/2} \langle 
e^{u^{\prime 
*}\hat{a}^\dagger} e^{v^{\prime *}\hat{b}^\dagger} e^{u'\hat{a}} e^{v'\hat{b}} \rangle,
\end{split}
\end{equation}
where $u = \pm v = t/\sqrt{2}$ and $u' = \pm v' = it/\sqrt{2}$. Using the representation of the 
squeezed state in the form \eqref{eq:sq-st}, one can compute the following quantity:
\begin{equation}
\begin{split}
    &{}_\lambda\langle 0| e^{u^*\hat{a}^\dagger} e^{v^*\hat{b}^\dagger} e^{u\hat{a}} e^{v\hat{b}}
    |0\rangle_\lambda \\
    &= (1-\lambda^2) \langle 0, 0| e^{\lambda \hat{a} \hat{b}} 
e^{u^*\hat{a}^\dagger} e^{v^*\hat{b}^\dagger} e^{u\hat{a}} e^{v\hat{b}} e^{\lambda \hat{a}^\dagger 
\hat{b}^\dagger} |0, 0\rangle.
\end{split}
\end{equation}
The product inside the brackets can be transformed with the help of the 
BCH relation
\begin{equation}\label{eq:AB}
    e^{\hat{A}} e^{\hat{B}} = e^{[\hat{A},\hat{B}]} e^{\hat{B}} e^{\hat{A}},
\end{equation}
which is valid if the commutator $[\hat{A},\hat{B}]$ commutes with both $\hat{A}$ and $\hat{B}$, 
and the equality
\begin{equation}\label{eq:ee}
    e^{\mu \hat{a} \hat{b}} e^{\nu \hat{a}^\dagger \hat{b}^\dagger} |0,0\rangle = 
    \frac{1}{1 - \mu \nu} \exp\left(\frac{\nu}{1 - \mu \nu}\hat{a}^\dagger \hat{b}^\dagger\right) 
|0, 0\rangle,
\end{equation}
which is derived in Ref.~\cite{*[][{ p. 187}] psqq}. Applying Eq.~\eqref{eq:AB} several times with, 
for example, $\hat{A} = \hat{a}$, $\hat{B} = \hat{a}^\dagger \hat{b}^\dagger$, $[\hat{A},\hat{B}] = 
\hat{b}^\dagger$, and finally using Eq.~\eqref{eq:ee}, we have
\begin{equation}
\begin{split}
    {}_\lambda\langle 0| &e^{u^*\hat{a}^\dagger} e^{v^*\hat{b}^\dagger} e^{u\hat{a}} e^{v\hat{b}}
    |0\rangle_\lambda \\
    &= \exp\left(\frac{\lambda}{1-\lambda^2}(uv + u^* v^* + \lambda |u|^2 + \lambda |v|^2)\right),
\end{split}
\end{equation}
from which we immediately obtain that
\begin{equation}
    \mathcal{X}_\pm(t) = \mathcal{P}_\mp(t) = \exp\left(\frac{1 \pm \lambda}{1 \mp \lambda} 
\frac{t^2}{2}\right).
\end{equation}
Expanding both sides in $t$ and comparing the coefficients we get the relations \eqref{eq:xp-sq}.
\end{proof}

The inequalities \eqref{eq:xpn-sep} can be used to demonstrate that some of the 
minimal states \eqref{eq:Psi-bip} are entangled. In fact, let us compute the 
left-hand side of the inequality \eqref{eq:xpn-sep} on a state of the form 
\eqref{eq:Psi-bip}. One can easily find that
\begin{equation}
    \langle (\hat{x}_a \pm \hat{x}_b)^{2n} + (\hat{p}_a \mp \hat{p}_b)^{2n} \rangle = 
    2^n 
    \begin{cases}
        \langle \hat{x}^{2n}_a \rangle + \langle \hat{p}^{2n}_b \rangle \\
        \langle \hat{x}^{2n}_b \rangle + \langle \hat{p}^{2n}_a \rangle,
    \end{cases}
\end{equation}
depending on the combinations of signs, where subscript $a$ means averaging over the state with 
wave function $\psi_{\mathrm{min}}(u)$ and subscript $b$ means averaging over the state with wave 
function $\varphi(v)$. According to Eq.~\eqref{eq:x2p}, we have 
\begin{equation}
    \langle (\hat{x}_a \pm \hat{x}_b)^{2n} + (\hat{p}_a \mp \hat{p}_b)^{2n} = 
\frac{\Lambda^{(2n)}_{\mathrm{min}}}{2} + 
    2^n
    \begin{cases}
        \langle \hat{x}^{2n}_b \rangle \\
        \langle \hat{p}^{2n}_b \rangle,
    \end{cases}
\end{equation}
and thus the states \eqref{eq:Psi-bip} violate the inequalities \eqref{eq:xpn-sep} depending on 
whether the state $\varphi(v)$ exhibits higher-order squeezing, i.e., whether 
one of the quantities $\langle \hat{x}^{2n}_b \rangle$ and $\langle 
\hat{p}^{2n}_b \rangle$ is less than $\lambda^{(2n)}_{\mathrm{min}}/2$.

We now prove that for the case of $n>1$, all bipartite minimal states are 
entangled, not only those states with a higher-order squeezing component.
\begin{thrm}
Any pure state with a wave function of the form \eqref{eq:Psi-bip} is entangled, 
provided that $\psi_{\mathrm{min}}(x)$ is the single-partite minimal solution of 
Eq.~\eqref{eq:eve} with $n>1$ and $\varphi(y)$ is an arbitrary wave function.
\end{thrm}
\begin{proof}
A bipartite pure state is separable only if it is factorizable, so we must prove 
that there are no functions $f(x)$ and $g(y)$ such that 
\begin{equation}\label{eq:fg}
    \psi_{\mathrm{min}}\left(\frac{x \pm y}{\sqrt{2}}\right) 
\varphi\left(\frac{x \mp y}{\sqrt{2}}\right) = f(x) g(y).
\end{equation}
First note that this theorem is not valid for arbitrary functions $\psi_{\mathrm{min}}(x)$ and 
$\varphi(y)$. In fact, if both $\psi_{\mathrm{min}}(x)$ and $\varphi(y)$ are wave functions of 
the vacuum state,
\begin{equation}
    \psi_{\mathrm{min}}(x) = \frac{1}{\sqrt[4]{\pi}} e^{-x^2/2}, \quad
    \varphi(y) = \frac{1}{\sqrt[4]{\pi}} e^{-y^2/2},
\end{equation}
then we have
\begin{equation}
    \psi_{\mathrm{min}}\left(\frac{x \pm y}{\sqrt{2}}\right) 
\varphi\left(\frac{x \mp y}{\sqrt{2}}\right) = \psi_{\mathrm{min}}(x) \varphi(y).
\end{equation}
The only property of $\psi_{\mathrm{min}}(x)$ that we need in order to 
establish the theorem is that it takes on negative values, which is the case 
for any $n>1$ according to our numerical analysis, and that 
$\psi_{\mathrm{min}}(x)$ does not have too many zeroes.

We consider only one combination of signs, the proof for the other one is 
similar. Let us assume that the relation \eqref{eq:fg} is valid for some 
functions $f(x)$ and $g(y)$. Then it is easy to derive the following identity:
\begin{equation}\label{eq:psi-phi}
\begin{split}
    \psi_{\mathrm{min}}(x) &\psi_{\mathrm{min}}(y) \varphi(x) \varphi(-y) \\
    &= \psi_{\mathrm{min}}(0) \varphi(0) \psi_{\mathrm{min}}(x+y) \varphi(x-y),
\end{split}
\end{equation}
which holds for all real numbers $x$ and $y$. If $\varphi(0) = 0$, then we must 
have $\psi_{\mathrm{min}}(x) \psi_{\mathrm{min}}(y) \varphi(x) \varphi(-y) = 0$ 
for all points $x$ and $y$. Since $\psi_{\mathrm{min}}(x)$ has at most 
countably many zeroes (we assume that $\psi_{\mathrm{min}}(x)$ behaves like the 
function \eqref{eq:BJ0} that accurately approximates it) and $\varphi(y)$ is 
normalized, there must by some $x_0$ such that both $\psi_{\mathrm{min}}(x_0)$ 
and $\varphi(x_0)$ are nonzero. Then we find that $\psi_{\mathrm{min}}(y) 
\varphi(-y) = 0$ for all $y$, and thus the norm of $\varphi(y)$ is zero --- a 
contradiction.

So, we have $\varphi(0) \not= 0$ and since the global sign of $\varphi$ is unimportant, we can 
assume that $\varphi(0) > 0$. Setting $y = -x$ in Eq.~\eqref{eq:psi-phi} and taking into 
account the symmetry of $\psi_{\mathrm{min}}(x)$, we obtain
\begin{equation}
    \psi_{\mathrm{min}}^2(x) \varphi^2(x) = \psi_{\mathrm{min}}^2(0) \varphi(0) \varphi(2x),
\end{equation}
from which we conclude that we must have $\varphi(x) \geq 0$ for all $x$, and 
so if $\varphi(x) < 0$ for some $x$ then $\Psi_{\mathrm{min}}(x, y)$ is not 
separable. If $\varphi(x) \geq 0$ for all $x$ then let us take a number 
$x_1$ such that $\psi_{\mathrm{min}}(x_1) < 0$. If we substitute $x = y = x_1/2$ 
into Eq.~\eqref{eq:psi-phi}, we have
\begin{equation}
\begin{split}
    0 &\leq \psi_{\mathrm{min}}^2(x_1/2) \varphi(x_1/2) \varphi(-x_1/2) \\
    &= \psi_{\mathrm{min}}(0) \varphi^2(0) \psi_{\mathrm{min}}(x_1) < 0,
\end{split}
\end{equation}
which is again a contradiction, proving the theorem. 
\end{proof}

\subsection{Approach two}

We now present another approach, which is better for getting the higher-order conditions, because in Approach one the 
second-order violations typically appear to come together with higher-order violations. In this new approach we derive 
conditions that can be violated only by non-Gaussian states. The inequality \eqref{eq:xp-2} can also be written in the 
following form:
\begin{equation}
    \langle (\hat{a}^\dagger \pm \hat{b})(\hat{a} \pm \hat{b}^\dagger) \rangle \geq 1.
\end{equation}
This form suggests another way to extend Eq.~\eqref{eq:xp-2}, by increasing the powers of the annihilation and creation 
operators. We first analyze the case of second powers (so the total product will be fourth-order) and derive a state 
that perfectly violates the resulting condition.
\begin{thrm}
For any bipartite separable state the following inequalities are valid:
\begin{subnumcases}{\label{eq:bS} \langle (\hat{a}^{\dagger 2} \pm \hat{b}^2)(\hat{a}^2 \pm 
\hat{b}^{\dagger 2}) \rangle \geq}
        2 & \text{for sep. states} \label{eq:bS:1} \\
        0 & \text{for all states} \label{eq:bS:2}
\end{subnumcases}
The inequality \eqref{eq:bS:2} is always strict and tight, i.e., the left-hand side can be arbitrarily close to zero 
(but never equal to zero).
\end{thrm}
\begin{proof}
We first prove the inequality \eqref{eq:bS:1}. For a partially transposed state we have
\begin{equation}
\begin{split}
    \langle (\hat{a}^{\dagger 2} &\pm \hat{b}^2)(\hat{a}^2 \pm 
\hat{b}^{\dagger 2}) \rangle^{\mathrm{PT}} \\
&= \langle (\hat{a}^{\dagger 2} \pm \hat{b}^{\dagger 2})(\hat{a}^2 
\pm \hat{b}^2) + 4\hat{b}^\dagger \hat{b} + 2\rangle \geq 2,
\end{split}
\end{equation}
and thus obtain the desired relation in the same way as we derived the inequality \eqref{eq:xpn-sep}. The lower bound is 
attained, for example, for the bipartite two-mode vacuum state.

Now we prove that the left-hand side of the inequalities \eqref{eq:bS} can never be equal to zero (for a physical 
state). It is enough to prove this statement for pure states only. If for a pure state $|\psi\rangle = 
\sum^{+\infty}_{k,l=0} c_{k,l} |k,l\rangle$ the left-hand side of \eqref{eq:bS} is zero then we must have $(\hat{a}^2 
\pm \hat{b}^{\dagger 2}) |\psi\rangle = 0$. From this we get the following relation between the coefficients of the 
state:
\begin{equation}
    \sqrt{(k+1)(k+2)} c_{k+2, l} = \mp \sqrt{l(l-1)} c_{k, l-2},
\end{equation}
which holds for all $k,l \geq 0$. We immediately find that $c_{k0} = c_{k1} = 0$ for all $k \geq 2$. The relation above 
can be rewritten as
\begin{equation}
    c_{k+2, l+2} = \mp \sqrt{\frac{(l+1)(l+2)}{(k+1)(k+2)}} c_{kl},
\end{equation}
for all $k, l \geq 0$. Applying it several times we obtain the equality
\begin{equation}\label{eq:c2}
    c_{k+2j, l+2j} = (\mp 1)^j \sqrt{\frac{(l+2j)!}{(k+2j)!} \frac{k!}{l!}} 
c_{kl},
\end{equation}
for all $j \geq 0$. Since the state $|\psi\rangle$ is normalized, at least one of the coefficients $c_{kl}$ is non-zero, 
let us say, $c_{k_0l_0} \not= 0$. If $l_0 \geq k_0$ then $c_{k_0+2j, l_0+2j}$ does not tend to zero as $j \to +\infty$, 
and thus $\sum^{+\infty}_{j=0} |c_{k_0+2j, l_0+2j}|^2$ cannot converge. So, we must have $k_0 > l_0$. Then, using 
Eq.~\eqref{eq:c2}, we also have $c_{k_0-2, l_0-2} \not= 0$ (provided that $k_0-2, l_0-2 \geq 0$). Repeating the process 
of subtracting 2 from both indices several times, we either arrive at $c_{10}$, or $c_{k0}$ or $c_{k1}$ with $k \geq 2$. 
As we have already seen, in the latter case we must have $c_{k0} = c_{k1} = 0$, so the only possibility is that $k_0 = 
2j_0 + 1$ and $l_0 = 2 j_0$ for some $j_0 \geq 0$. We conclude that $c_{10} \not= 0$ and from Eq.~\eqref{eq:c2} we get
\begin{equation}
    c_{2j+1, 2j} = (\mp 1)^j \frac{1}{\sqrt{2j+1}} c_{10},
\end{equation}
for all $j \geq 0$. Then the following series must converge:
\begin{equation}
    \sum^{+\infty}_{j=0} |c_{2j+1, 2j}|^2 = |c_{10}|^2 \sum^{+\infty}_{j=0} \frac{1}{2j+1},
\end{equation}
which, however, is well-known to diverge. We have finally arrived at a contradiction: from the assumption $(\hat{a}^2 
\pm \hat{b}^{\dagger 2}) |\psi\rangle = 0$ we find that all coefficients $c_{kl}$ must be equal to zero, which 
contradicts the normalization of $|\psi\rangle$.

\begin{figure*}
\includegraphics{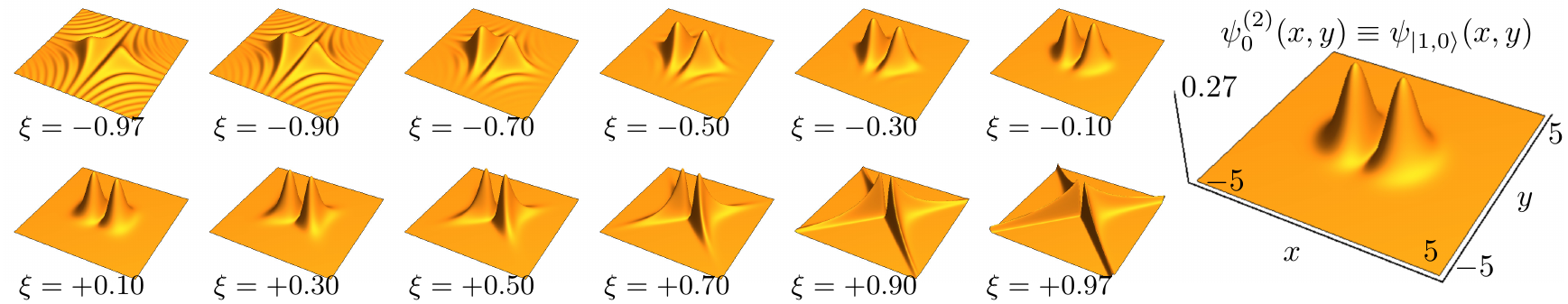}
\caption{The wave function $\psi^{(2)}_\xi(x,y)$ of the state \eqref{eq:psi2} for different values of $\xi$. The scale 
of each axis is shown on the right.}
\end{figure*}

To show that the left-hand side of the inequalities \eqref{eq:bS} can become arbitrarily small, consider the following 
state:
\begin{equation}\label{eq:psi2}
    |\psi^{(2)}_\xi\rangle = \sqrt{\frac{\xi}{\atanh\xi}} \sum^{+\infty}_{k=0} 
    \frac{\xi^k}{\sqrt{2k+1}} |2k+1, 2k\rangle.
\end{equation}
This state is defined for $|\xi| < 1$. For $\xi = 0$ it is just the factorizable state $|1, 0\rangle$. It is 
straightforward to compute the quantity in question on the state \eqref{eq:psi2}. We have
\begin{equation}
    \langle (\hat{a}^{\dagger 2} \pm \hat{b}^2)(\hat{a}^2 \pm \hat{b}^{\dagger 2}) \rangle = 
    \frac{2\xi}{(1\mp\xi)^2 \atanh\xi}.
\end{equation}
We see that $\langle (\hat{a}^{\dagger 2} \pm \hat{b}^2)(\hat{a}^2 \pm \hat{b}^{\dagger 2}) \rangle \to 0$ when $\xi \to 
\mp 1$ and thus, for $\xi$ sufficiently close to $\pm 1$, the left-hand side of Eq.~\eqref{eq:bS} can be made 
arbitrarily close to zero. Note that the inequality $\langle (\hat{a}^{\dagger 2} + \hat{b}^2)(\hat{a}^2 + 
\hat{b}^{\dagger 2}) \rangle < 2$ holds true for all $-1 < \xi < 0$ and $\langle (\hat{a}^{\dagger 2} - 
\hat{b}^2)(\hat{a}^2 - \hat{b}^{\dagger 2}) \rangle < 2$ holds true for $0 < \xi < 1$.
\end{proof}

The wave function of the state \eqref{eq:psi2} can be found explicitly. Even though this state looks a bit similar to 
the two-mode squeezed vacuum state \eqref{eq:sq-vac}, it is actually non-Gaussian. The exact form of its wave function 
and the full details of the derivation are given in Appendix C. Another state that perfectly violates inequality 
\eqref{eq:bS} is presented in Appendix D.

Since the inequalities \eqref{eq:bS} are based on the commutator properties of the creation and annihilation operators, 
arbitrary unitary transformations $\hat{a} \to \hat{U}^\dagger \hat{a} \hat{U}$, $\hat{b} \to \hat{V}^\dagger \hat{b} 
\hat{V}$ preserve these inequalities. As a special case, we can write the same inequalities \eqref{eq:bS} for $\Delta 
\hat{a}$ and $\Delta \hat{b}$ instead of $\hat{a}$ and $\hat{b}$ (where $\Delta \hat{A} = \hat{A} - \langle \hat{A} 
\rangle$). Moreover, as we show now, these inequalities can never be violated by bipartite Gaussian states. 

\begin{thrm}\label{thrm:bip}
All bipartite separable states and all bipartite (including inseparable) Gaussian states satisfy the following 
inequality:
\begin{equation}\label{eq:dbS} 
    \langle ((\Delta\hat{a})^{\dagger 2} \pm (\Delta\hat{b})^2)((\Delta\hat{a})^2 \pm 
    (\Delta\hat{b})^{\dagger 2}) \rangle \geq 2.
\end{equation}
The left-hand side of this inequality can be arbitrarily close to zero.
\end{thrm}

This theorem, which we prove in Appendix D, can be generalized for arbitrary orders, though it is more difficult to 
establish similar results analytically. Below we formulate the general theorem and give an example of a state that 
violates the corresponding separability condition. It seems that a stronger statement (a full analogue of the preceding 
theorem) is valid, but we are not able to present a strict mathematical proof of it.

\begin{thrm}
For any bipartite separable state and for any positive integer $n$ the following inequalities are valid:
\begin{subequations} \label{eq:nS}
\begin{empheq}[left={\langle (\hat{a}^{\dagger n} \pm \hat{b}^n)(\hat{a}^n \pm 
\hat{b}^{\dagger n}) \rangle \geq}\empheqlbrace]{alignat=2}
  & n! && \quad \text{for sep. states} \label{eq:nS:1} \\
  & 0 && \quad \text{for all states} \label{eq:nS:2}
\end{empheq}
\end{subequations}
There are states which violate the inequality \eqref{eq:nS:1} at least by a factor of 4.
\end{thrm}
\begin{proof} Consider the following state:
\begin{equation}
    |\psi^{(n)}_\xi\rangle = N_n(\xi) \sum^{+\infty}_{k=0} \frac{\xi^k}{\sqrt{\prod^{n-1}_{j=1}(nk+j)}} 
    |nk+n-1, nk\rangle. \nonumber
\end{equation}
We choose the normalization such that $N_n(\xi) > 0$. Note that for the case $n=2$ this definition coincides with 
Eq.~\eqref{eq:psi2}. The normalization factor is determined from the relation
\begin{equation}\label{eq:N-n}
    N^2_n(\xi) \sum^{+\infty}_{k=0} \frac{\xi^{2k}}{\prod^{n-1}_{j=1}(nk+j)} = 1. 
\end{equation}
Since $n \geq 3$, the series in this expression converges for $|\xi| \leq 1$, so the state now  is defined for the 
end-points $\xi=\pm 1$ of the interval of definition. From the relation above it follows that $N_n = N_n(+1) = N_n(-1)$ 
is a well defined number.

Computing the left-hand side of Eq.~\eqref{eq:nS} for the state $|\psi^{(n)}_\xi\rangle$, we obtain
\begin{equation}
    \langle (\hat{a}^{\dagger n} \pm \hat{b}^n)(\hat{a}^n \pm \hat{b}^{\dagger n}) \rangle = 
    N^2_n(\xi) \frac{n}{(1 \mp \xi)^2}.
\end{equation}
When $\xi \to \mp 1$, this quantity tends to 
\begin{equation}
    \lim_{\xi \to \mp 1} \langle (\hat{a}^{\dagger n} \pm \hat{b}^n)(\hat{a}^n \pm \hat{b}^{\dagger n}) \rangle 
    = N^2_n \frac{n}{4}.
\end{equation}
From the definition \eqref{eq:N-n} we have the following inequality for the norm:
\begin{equation}
    N^2_n = \left(\frac{1}{(n-1)!}+\ldots\right)^{-1} < (n-1)!.
\end{equation}
Using this inequality we can estimate the limiting values of the left-hand side of Eq.~\eqref{eq:nS},
\begin{equation}
    N^2_n \frac{n}{4} < \frac{n!}{4}.
\end{equation}
This shows that in the limit $\xi \to \pm 1$ the state $|\psi^{(n)}_\xi\rangle$ violates the inequality \eqref{eq:nS:1} 
at least by a factor of 4.
\end{proof}

Our conjecture is that the inequality \eqref{eq:nS:1} cannot be violated by Gaussian states (including all inseparable 
ones) and the inequality \eqref{eq:nS:2} is tight. We verified this by numerical computations, but we could not find 
rigorous proofs. 

\section{Conclusion}

Considering the well-known Duan separability condition in terms of second moments of quadrature operators, we presented 
a generalization of such conditions to higher orders in two different ways. One is expressed in terms of higher-order 
correlations between bipartite quadrature operators, which is very intuitive, resembling EPR-type correlations for 
higher orders, and which is rather convenient for the use in experiments. The other way, described by creation and 
annihilation operators and their higher-order combinations, leads to truly higher-order conditions, but it is more 
difficult to test experimentally. The resulting criteria in this case represent a true hierarchy, where certain 
higher-order inseparability conditions cannot be fulfilled by any entangled Gaussian states. Nonetheless, entangled 
non-Gaussian states exist that satisfy such conditions. Our approach can open new directions in the study of 
higher-order entanglement phenomena.

\appendix

\section{Computation of eigenvalues}

The minimal value of $\langle\hat{x}^{2n} + \hat{p}^{2n}\rangle$ is computed numerically as the minimal eigenvalue of 
the truncated matrix $M_{2n}$. This eigenvalue quickly stabilizes as the order $N$ of truncation grows, so the truncated 
matrix even for small values of $N \approx 10^2$ gives a very precise result. But since the solutions for $n>1$ are 
nearly indistinguishable, it makes sense to compute the eigenvalues and, more importantly, the corresponding 
eigenvectors as precise as possible. 

To do it, the eigenvalues and eigenvectors are computed with Intel Math Kernel 
Library\footnote{\texttt{https://software.intel.com/en-us/intel-mkl}} in two steps. First, the truncated matrices 
$M_{2n}$ are treated as dense matrices and their minimal eigenvalues are computed with the routine \texttt{syevx}. The 
order of truncation $N \approx 10^3$ is typically used for this step. From the structure of the matrices $M_{2n}$, 
discussed in the main part of this work, it follows that these matrices are sparse and their nonzero elements constitute 
a small fraction of the total size of the matrices. This can be illustrated by Eq.~\eqref{eq:M4}, for example. For such 
matrices using a sparse solver is much more space efficient then using the standard dense solver. We use the sparse 
eigensolver \texttt{dfeast\_scsrev}, which is based on the FEAST algorithm proposed in Ref.~\cite{PhysRevB.79.115112}. 
The results obtained in the previous step with dense matrices are used as initial conditions for this sparse solver.

\begin{table*}[ht]
\begin{tabular}{l|D{.}{.}{4.8} D{.}{.}{4.8} 
D{.}{.}{4.8}D{.}{.}{4.8}D{.}{.}{4.8}D{.}{.}{4.8}D{.}{.}{5.8}}
\toprule[0.6pt] \\[-3mm]
$n$ & \multicolumn{1}{c}{\hfil $\lambda^{(n)}_{\mathrm{min}}$} & \multicolumn{1}{c}{\hfil $\psi(0)$} & 
\multicolumn{1}{c}{\hfil $\psi^{\prime\prime}(0)$} & \multicolumn{1}{c}{\hfil $\psi^{(4)}(0)$} & 
\multicolumn{1}{c}{\hfil $\psi^{(6)}(0)$} & \multicolumn{1}{c}{\hfil $\psi^{(8)}(0)$} & 
\multicolumn{1}{c}{\hfil $\psi^{(10)}(0)$} \\[1mm]
\midrule[0.4pt] \\[-2mm]
1 & 1 & 0.75112554 & \multicolumn{1}{c}{\hfil *} & \multicolumn{1}{c}{\hfil *} & \multicolumn{1}{c}{\hfil *} & 
\multicolumn{1}{c}{\hfil *} & \multicolumn{1}{c}{\hfil *} \\[1mm]
2 & 1.39672823 & 0.73253810 & -0.59978918 & \multicolumn{1}{c}{\hfil *} & \multicolumn{1}{c}{\hfil *} & 
\multicolumn{1}{c}{\hfil *} & \multicolumn{1}{c}{\hfil *} \\[1mm]
3 & 2.95304540 & 0.73255327 & -0.60402445 & 1.10905904 & \multicolumn{1}{c}{\hfil *} & 
\multicolumn{1}{c}{\hfil *} & \multicolumn{1}{c}{\hfil *} \\[1mm]
4 & 8.28911703 & 0.73460748 & -0.61951231 & 1.22274755 & -2.94050192 & \multicolumn{1}{c}{\hfil *} & 
\multicolumn{1}{c}{\hfil *} \\[1mm]
5 & 28.97408955 & 0.73662780 & -0.63430030 & 1.32592056 & -3.61002601 & 9.96484721 & \multicolumn{1}{c}{\hfil *} \\[1mm]
6 & 121.21680669 & 0.73832570 & -0.64684279 & 1.41413494 & -4.19672348 & 13.62188458 & -40.89217482 \\[2mm]
\bottomrule[0.6pt]
\end{tabular}
\caption{The minimal eigenvalues and the corresponding values for the derivatives.}\label{tbl:2}
\end{table*}

The use of the dense solver is straightforward. The sparse solver is slightly more tricky to use because it requires 
more efforts to prepare the matrices in the format required by this solver. Nevertheless, these extra efforts pay off 
since the order of truncation can be increased to $N \approx 10^6$ on the same hardware. Unfortunately, further increase 
of $N$ leads to unstable behavior of the solver --- the solution starts to depend on the number of cores used for the 
computation, so one needs an independent way to verify the results. One way to do it is to use these results to set the 
initial conditions $\psi(0)$, $\psi'(0)$, \ldots, for the ODE \eqref{eq:eve}, solve it and compare the two solutions.

Having the eigenvectors $(c_0, c_1, \ldots)$, we can compute the corresponding wave function \eqref{eq:psi-c} with the 
help of the recurrence relation
\begin{equation}
    \psi_{k+1}(x) = \sqrt{\frac{2}{k+1}} x \psi_k(x) - \sqrt{\frac{k}{k+1}} 
\psi_{k-1}(x)
\end{equation}
with the initial condition
\begin{equation}
    \psi_0(x) = \frac{1}{\sqrt[4]{\pi}} e^{-x^2/2}, \quad 
    \psi_1(x) = \frac{\sqrt{2}}{\sqrt[4]{\pi}} x e^{-x^2/2}.
\end{equation}
In this way we can compute $\psi(0)$ and get the first initial condition for the equation \eqref{eq:eve}. To compute the 
derivatives note that
\begin{equation}
    \psi'(x) = \sum^{+\infty}_{k=0} c'_k \psi_k(x),
\end{equation}
where the new coefficients $c'_k$ are given by
\begin{equation}
    c'_k = \frac{1}{\sqrt{2}} (\sqrt{k+1} c_{k+1} - \sqrt{k} c_{k-1}).
\end{equation}
This means that the derivatives $\psi'(0)$, $\psi^{\prime\prime}(0)$ and so on, can be computed by the same routine used 
to compute the wave function itself, provided that this routine is given the new coefficients as input. The results of 
these computations are presented in the Table~\ref{tbl:2}. Note that our results agree (and greatly extend) those of 
Ref.~\cite{JMP-31-1947}.

\section{Alternate approach based on higher-order uncertainties}

A function $f(x)$ is referred to as convex if 
\begin{equation}\label{eq:convex}
    f(t x_1 + (1-t)x_2) \leq t f(x_1) + (1-t) f(x_2)
\end{equation}
for all $x_1$, $x_2$ and all $0 \leq t \leq 1$. If reverse inequality is valid,
\begin{equation}\label{eq:concave}
    f(t x_1 + (1-t)x_2) \geq t f(x_1) + (1-t) f(x_2)
\end{equation}
for all $x_1$, $x_2$ and all $0 \leq t \leq 1$, then the function $f$ is referred to as concave. The parameter $x$ is 
not necessarily to be a number, it can be a quantum state (and whatever else for what convex combinations are defined). 
In particular, a linear function is both convex and concave. If, for a given convex function $f$ of multipartite quantum 
state, we establish the inequality 
\begin{equation}\label{eq:fmax}
   f(\hat{\varrho}) \leq f_{\mathrm{max}}
\end{equation}
for all factorizable states then, due to inequality \eqref{eq:convex}, we automatically obtain that this inequality is 
valid for all separable states. If for a concave function we establish the inequality 
\begin{equation}\label{eq:fmin}
    f(\hat{\varrho}) \geq f_{\mathrm{min}}
\end{equation}
for all factorizable states then, due to inequality \eqref{eq:concave}, this inequality is automatically valid for all 
separable states. The inequalities of the form \eqref{eq:fmax} are more typical for finite-dimensional quantum systems, 
where the quantities under study are bounded from above (like Bell inequalities). The inequalities of the form 
\eqref{eq:fmin} are more natural in continuous-variable case, where the quantities can be arbitrarily large but bounded 
from below (like uncertainty relation). In both cases it is not necessary to explicitly keep track of all factorizable 
components of a separable state, since if the quantity in question has some convexity property then it is enough to 
consider only factorizable states, which greatly simplifies the notation.

Since $\langle \hat{A} \rangle_{\hat{\varrho}}$ is a linear function of $\hat{\varrho}$ it is both linear and concave, 
so it is enough to consider factorizable states only. For a factorizable state we have 
\begin{equation}
\begin{split}
    \langle (\hat{x}_a &+ \hat{x}_b)^4 + (\hat{p}_a - \hat{p}_b)^4 \rangle = 
    \langle \hat{x}^4_a + \hat{p}^4_a \rangle + \langle \hat{x}^4_b + \hat{p}^4_b \rangle \\
    &+ 6 \langle \hat{x}^2_a \rangle \langle \hat{x}^2_b\rangle + 6 \langle \hat{p}^2_a \rangle \langle 
    \hat{p}^2_b\rangle.
\end{split}
\end{equation}
The right-hand side of these equalities can be easily estimated as
\begin{equation}
\begin{split}
    \langle (\hat{x}_a &+ \hat{x}_b)^4 + (\hat{p}_a - \hat{p}_b)^4 \rangle \\
    &\geq 2 \lambda^{(4)}_{\mathrm{min}} + 12 \sqrt{\langle \hat{x}^2_a 
    \rangle \langle \hat{p}^2_a \rangle \langle \hat{x}^2_b \rangle \langle \hat{p}^2_b \rangle} \\
    &\geq 2 \lambda^{(4)}_{\mathrm{min}} + 3 = 5.7934 > \Lambda^{(4)}_{\mathrm{min}} = 5.5868,
\end{split}
\end{equation}
and thus the inequality $\langle (\hat{x}_a + \hat{x}_b)^4 + (\hat{p}_a - \hat{p}_b)^4 \rangle \geq 5.5868$ is valid for 
all bipartite separable states. We see that in the case of $2n=4$ this approach gives slightly better estimation than 
the method based on partial transposition. 

This approach can be extended for higher orders, but for $2n>4$ it will give worse results. For example, for $2n=6$ we 
have
\begin{equation}
\begin{split}
    &\langle (\hat{x}_a + \hat{x}_b)^6 + (\hat{p}_a - \hat{p}_b)^6 \rangle = 
    \langle \hat{x}^6_a + \hat{p}^6_a + \hat{x}^6_b + \hat{p}^6_b \rangle \\
    &+ 15 (\langle \hat{x}^2_a \rangle \langle \hat{x}^4_b \rangle + \langle 
    \hat{x}^4_a \rangle \langle \hat{x}^2_b \rangle + \langle \hat{p}^2_a \rangle 
    \langle \hat{p}^4_b \rangle + \langle \hat{p}^4_a \rangle \langle \hat{p}^2_b \rangle) \\
    &+ 20 (\langle \hat{x}^3_a \rangle \langle \hat{x}^3_b \rangle - \langle \hat{p}^3_a \rangle \langle \hat{p}^3_b 
    \rangle).
\end{split}
\end{equation}
The moments of third order can be estimated as $|\langle \hat{x}^3 \rangle| \leq \sqrt{\langle \hat{x}^2 \rangle \langle 
\hat{x}^4 \rangle}$, and we can write
\begin{equation}\label{eq:xp6}
\begin{split}
    &\langle (\hat{x}_a + \hat{x}_b)^6 + (\hat{p}_a - \hat{p}_b)^6 \rangle 
    \geq \langle \hat{x}^6_a + \hat{p}^6_a + \hat{x}^6_b + \hat{p}^6_b \rangle \\ 
    &+ 15 (\langle \hat{x}^2_a \rangle \langle \hat{x}^4_b \rangle + \langle 
    \hat{x}^4_a \rangle \langle \hat{x}^2_b \rangle + \langle \hat{p}^2_a \rangle 
    \langle \hat{p}^4_b \rangle + \langle \hat{p}^4_a \rangle \langle \hat{p}^2_b \rangle) \\
    &- 20 \Bigl(\sqrt{\langle \hat{x}^2_a \rangle \langle \hat{x}^4_a 
    \rangle \langle \hat{x}^2_b \rangle \langle \hat{x}^4_b \rangle} + 
    \sqrt{\langle \hat{p}^2_a \rangle \langle \hat{p}^4_a \rangle \langle 
    \hat{p}^2_b \rangle \langle \hat{p}^4_b \rangle}\Bigr).
\end{split}
\end{equation}
We can combine the terms on the right-hand side in such a way to get full squares so that we have
\begin{equation}
\begin{split}
    &\text{rhs of Eq.~\eqref{eq:xp6}}\ = \langle \hat{x}^6_a + \hat{p}^6_a \rangle + \langle \hat{x}^6_b + \hat{p}^6_b 
    \rangle \\
    &+ 5 (\langle \hat{x}^2_a \rangle \langle \hat{x}^4_b \rangle + 
    \langle \hat{x}^4_a \rangle \langle \hat{x}^2_b \rangle + \langle \hat{p}^2_a \rangle 
    \langle \hat{p}^4_b \rangle + \langle \hat{p}^4_a \rangle \langle \hat{p}^2_b \rangle) \\
    &+ 10 \Bigr(\sqrt{\langle \hat{x}^2_a \rangle \langle \hat{x}^4_b \rangle} 
    - \sqrt{\langle \hat{x}^4_a \rangle \langle \hat{x}^2_b \rangle}\Bigr)^2 \\
    &+ 10 \Bigl(\sqrt{\langle \hat{p}^2_a \rangle \langle \hat{p}^4_b \rangle} 
    - \sqrt{\langle \hat{p}^4_a \rangle \langle \hat{p}^2_b \rangle}\Bigr)^2. 
\end{split}
\end{equation}
If we ignore the squares and apply the inequality $a+b \geq 2\sqrt{ab}$ twice, we obtain
\begin{equation}
\begin{split}
    &\langle (\hat{x}_a + \hat{x}_b)^6 + (\hat{p}_a - \hat{p}_b)^6 \rangle \geq 2 \lambda^{(3)}_{\mathrm{min}} \\
    &+ 10\Bigl(\sqrt{\langle \hat{x}^2_a \rangle \langle \hat{p}^2_a \rangle 
    \langle \hat{x}^4_b \rangle \langle \hat{p}^4_b \rangle} + \sqrt{\langle 
    \hat{x}^4_a \rangle \langle \hat{p}^4_a \rangle \langle \hat{x}^2_b \rangle \langle \hat{p}^2_b \rangle}\Bigr).
\end{split}
\end{equation}
To move further, we need a higher-order analog of the Heisenberg uncertainty relation. This relation has been derived in 
Ref.~\cite{JPhysA.33.L83} and reads as 
\begin{equation}
    \langle \hat{x}^{2n} \rangle \langle \hat{p}^{2n} \rangle \geq \left(\frac{(2n)!}{2^{2n}n!}\right)^2,
\end{equation}
where the number on the right-hand side is the value of the left-hand side at vacuum. This leads to the estimation
\begin{equation}
\begin{split}
    \langle (\hat{x}_a &+ \hat{x}_b)^6 + (\hat{p}_a - \hat{p}_b)^6 \rangle \\
    &\geq 2 \lambda^{(3)}_{\mathrm{min}} + 10\left(\sqrt{\frac{1}{4}\frac{9}{16}} + 
    \sqrt{\frac{1}{4}\frac{9}{16}}\right) \\ 
    &= 2 \lambda^{(3)}_{\mathrm{min}} + \frac{30}{4} = 13.406,
\end{split}
\end{equation}
which is weaker then the bound obtained with PPT. The conclusion --- one cannot simply estimate this quantity term by 
term. There are several sources of loosing information --- we replace each term (more precisely, each group of terms) 
with simpler terms, then we estimate them individually. The replacement may not be tight, and individual estimation of 
terms is definitely not tight since different terms take minimum at different states. 

\section{The wave function of the state \eqref{eq:psi2}}

\begin{thrm}
The wave function of the state \eqref{eq:psi2} is given by the following expressions:
\begin{equation}\label{eq:bS2+}
\begin{split}
    \psi^{(2)}_\xi(x, y) & = \frac{1}{2} 
\frac{\exp\left(\frac{y^2-x^2}{2}\right)}{\sqrt{2\atanh\xi}} \\
&\biggl[\erf\left(\frac{y+x\sqrt{\xi}}{\sqrt{1-\xi}}\right) - 
\erf\left(\frac{y-x\sqrt{\xi}}{\sqrt{1-\xi}}\right)\biggr]
\end{split}
\end{equation}
if $0 \leq \xi < 1$ and 
\begin{equation}\label{eq:bS2-}
\begin{split}
    \psi^{(2)}_\xi(x, y) & = 
\frac{\exp\left(\frac{y^2-x^2}{2}\right)}{\sqrt{-2\atanh\xi}} 
\im\left[\erf\left(\frac{y + i\sqrt{-\xi}x}{\sqrt{1-\xi}}\right)\right]
\end{split}
\end{equation}
if $-1 < \xi \leq 0$.
\end{thrm}
\begin{proof}
Note that if we let $\xi$ be negative in Eq.~\eqref{eq:bS2+} we get 
Eq.~\eqref{eq:bS2-} with 
uncertainty in sign (since there are two complex square roots of a negative 
number which differ 
in sign), so Eq.~\eqref{eq:bS2-} gives the correct expression in the case of 
negative $\xi$. From 
the definition \eqref{eq:psi2} we have
\begin{equation}\label{eq:psi-1}
    \psi^{(2)}_\xi(x, y) = \sqrt{\frac{\xi}{2\pi\atanh\xi}} S_\xi(x, y) 
e^{-\frac{x^2+y^2}{2}},
\end{equation}
where $S_\xi(x, y)$ is given by the following series:
\begin{equation}
    S_\xi(x, y) = \sum^{+\infty}_{k=0} \frac{\xi^k}{2^{2k}(2k+1)!} H_{2k+1}(x) 
H_{2k}(y).
\end{equation}
We first derive the expression for this series for $0 \leq \xi < 1$. To do it, 
let us take the 
partial derivative with respect to $x$. We get
\begin{equation}\label{eq:pS}
    \frac{\partial S_\xi}{\partial x} = 2 \sum^{+\infty}_{k=0} \frac{\xi^k}{2^{2k}(2k)!} H_{2k}(x) 
H_{2k}(y) = s(\eta) + s(-\eta),
\end{equation}
where $\eta = \sqrt{\xi}$, and $s(\eta)$ is defined via
\begin{equation}\label{eq:f}
\begin{split}
    s(\eta) &= \sum^{+\infty}_{k=0} \frac{\eta^k}{2^k k!} H_k(x) H_k(y) \\
    &= \frac{1}{\sqrt{1-\eta^2}} \exp\left(\frac{2xy\eta - (x^2+y^2)\eta^2}{1-\eta^2}\right).
\end{split}
\end{equation}
We thus obtain the following expression for the partial derivative:
\begin{equation}
    \frac{\partial S_\xi}{\partial x} = \frac{2}{\sqrt{1-\xi}} 
\cosh\left(\frac{2xy\sqrt{\xi}}{1-\xi}\right) \exp\left(-\frac{(x^2+y^2)\xi}{1-\xi}\right). 
\nonumber
\end{equation}
Integrating both sides of this relation and taking into account that $S_\xi(0, y) = 0$ we arrive to 
an expression for the series $S_\xi(x, y)$
\begin{equation}
    S_\xi(x, y) = \sqrt{\frac{\pi}{\xi}}\frac{e^{y^2}}{2} 
\biggl[\erf\left(\frac{y+x\sqrt{\xi}}{\sqrt{1-\xi}}\right) - 
\erf\left(\frac{y-x\sqrt{\xi}}{\sqrt{1-\xi}}\right)\biggr]. \nonumber
\end{equation}
Combining it with Eq.~\eqref{eq:psi-1}, we get the wave function given by Eq.~\eqref{eq:bS2+}.

In the case of negative $\xi$ the relation \eqref{eq:pS} is valid provided that $\eta = 
i\sqrt{-\xi}$. We have the following expression for the partial derivative:
\begin{equation}
    \frac{\partial S_\xi}{\partial x} = \frac{2}{\sqrt{1-\xi}} 
    \cos\left(\frac{2xy\sqrt{-\xi}}{1-\xi}\right) \exp\left(-\frac{(x^2+y^2)\xi}{1-\xi}\right). \nonumber
\end{equation}
Integrating and taking into account that $S_\xi(0, y) = 0$, we obtain 
\begin{equation}
    S_\xi(x, y) = \sqrt{-\frac{\pi}{\xi}} e^{y^2} \im\left[\erf\left(\frac{y+i\sqrt{-\xi}x}{\sqrt{1-\xi}}\right)\right],
\end{equation}
which leads to the wave function \eqref{eq:bS2-}.

For $\xi=0$ the wave functions \eqref{eq:bS2+} and \eqref{eq:bS2-} must be the wave function of the state $|1,0\rangle$. 
Direct substitution $\xi=0$ into those expressions results in the indeterminate form $0/0$, so more careful analysis is 
needed to determine the value of the functions \eqref{eq:bS2+} and \eqref{eq:bS2-} at $\xi=0$ and demonstrate that it is 
exactly the wave function of the state $|1,0\rangle$. We show that this is true in the limit $\xi \to 0$. We first 
consider the case of $\xi$ approaching zero from above. We have 
\begin{equation}
    \erf\left(\frac{y\pm x\sqrt{\xi}}{\sqrt{1-\xi}}\right) = \erf(y) - 
    \frac{2e^{-y^2}}{\sqrt{\pi}} \left(y-\frac{y\pm x\sqrt{\xi}}{\sqrt{1-\xi}}\right) + \mathcal{O}(\xi), \nonumber
\end{equation}
which follows from the relations
\begin{equation}
    \erf'(y) = \frac{2}{\sqrt{\pi}} e^{-y^2}, \quad y-\frac{y\pm x\sqrt{\xi}}{\sqrt{1-\xi}} = \mathcal{O}(\sqrt{\xi}),
\end{equation}
and the Taylor expansion of the error function at $y$. Subtracting one from the other, we get
\begin{equation}
\begin{split}
    \erf\left(\frac{y+x\sqrt{\xi}}{\sqrt{1-\xi}}\right) &- \erf\left(\frac{y-x\sqrt{\xi}}{\sqrt{1-\xi}}\right) \\
    &= \frac{4}{\sqrt{\pi}} e^{-y^2} \frac{x\sqrt{\xi}}{\sqrt{1-\xi}} + \mathcal{O}(\xi).
\end{split}
\end{equation}
Substituting this to Eq.~\eqref{eq:bS2+} and taking the limit $\xi \to 0$ we get $\sqrt{2/\pi} x 
e^{-\frac{x^2+y^2}{2}}$, i.e., the wave function of the state $|1, 0\rangle$. The case of $\xi$ approaching zero from 
below can be considered analogously.
\end{proof}

\section{Another non-Gaussian state}

Here we present another state that maximally violates the inequality \eqref{eq:bS}. It reads as follows:
\begin{equation}\label{eq:psi2'}
    |\psi^{\prime(2)}_\xi\rangle = \sqrt{\frac{-2\xi^2}{\ln(1-\xi^2)}} \sum^{+\infty}_{k=0} 
    \frac{\xi^k}{\sqrt{2k+2}} |2k+2, 2k+1\rangle,
\end{equation}
where $|\xi| < 1$. For $\xi = 0$ it is just the factorizable state $|2, 1\rangle$. The left-hand side of 
Eq.~\eqref{eq:bS} for this state is given by
\begin{equation}
    \langle (\hat{a}^{\dagger 2} \pm \hat{b}^2)(\hat{a}^2 \pm \hat{b}^{\dagger 2}) \rangle =
    -\frac{4\xi^2 (2\mp\xi)}{\ln(1-\xi^2)(1\mp\xi)^2},
\end{equation}
and we see that $\langle (\hat{a}^{\dagger 2} \pm \hat{b}^2)(\hat{a}^2 \pm \hat{b}^{\dagger 2}) \rangle \to 0$ when $\xi 
\to \mp 1$.

\begin{figure*}
\includegraphics{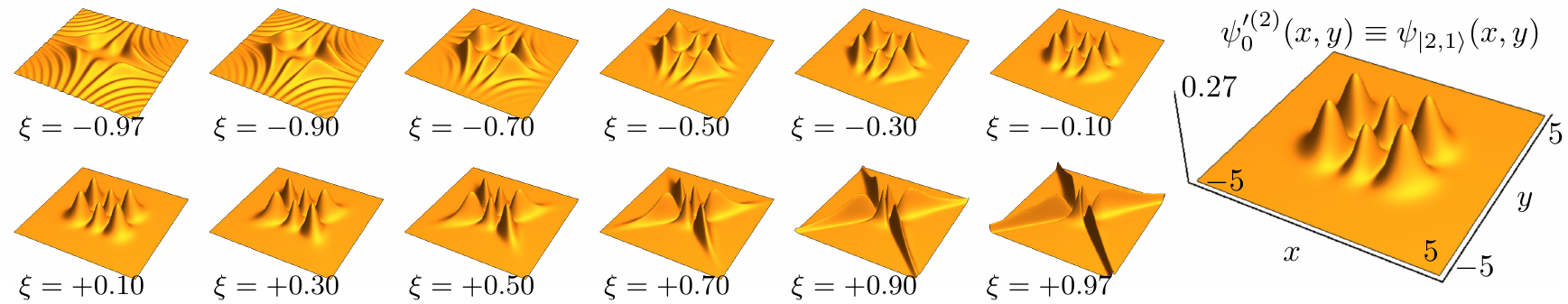}
\caption{The wave function $\psi^{\prime (2)}_\xi(x,y)$ of the state \eqref{eq:psi2'} for different values of $\xi$. The 
scale of each axis is shown on the right.}
\end{figure*}

We prove now that the wave function of the state \eqref{eq:psi2'} is given by
\begin{equation}\label{eq:bS'-1}
\begin{split}
    \psi^{\prime (2)}_\xi(x, y) &= \frac{\exp\left(\frac{y^2-x^2}{2}\right)}{2\sqrt{-\ln(1-\xi^2)}} \biggl[ 2\erf(y)  \\
    &+\erf\left(\frac{-y+\sqrt{\xi}x}{\sqrt{1-\xi}}\right) - \erf\left(\frac{y+\sqrt{\xi}x}{\sqrt{1-\xi}}\right)\biggr]
\end{split}
\end{equation}
for $0 < \xi < 1$, and by
\begin{equation}\label{eq:bS'-2}
\begin{split}
    \psi^{\prime (2)}_\xi(x, y) &= \frac{\exp\left(\frac{y^2-x^2}{2}\right)}{\sqrt{-\ln(1-\xi^2)}} \times \\
    &\left(\re\left[\erf\left(\frac{i\sqrt{-\xi}x + y}{\sqrt{1-\xi}}\right)\right] - \erf(y)\right).
\end{split}
\end{equation}
for $-1 < \xi < 0$. For $\xi = 0$ we get the wave function of the state $|2, 1\rangle$, i.e.,
\begin{equation}\label{eq:bS'-3}
    \psi^{\prime (2)}_0(x, y) = \frac{1}{\sqrt{\pi}} (2x^2-1) y e^{-\frac{x^2+y^2}{2}}.
\end{equation}
From the definition \eqref{eq:psi2'} we have
\begin{equation}
    \psi^{\prime (2)}_\xi(x, y) = \sqrt{\frac{-2\xi^2}{2\pi\ln(1-\xi^2)}} S_\xi(x, y) e^{-\frac{x^2+y^2}{2}},
\end{equation}
where $S_\xi(x, y)$ is given by the following series:
\begin{equation}
    S_\xi(x, y) = \sum^{+\infty}_{k=0} \frac{\xi^k}{2^{2k+1} (2k+2)!} H_{2k+2}(x) H_{2k+1}(y).
\end{equation}
A compact expression for this series can be obtained with the same trick that we used before ---  
by taking the partial derivative with respect to $x$. We have
\begin{equation}
    \frac{\partial S_\xi}{\partial x} = 2 \sum^{+\infty}_{k=0} \frac{\xi^k}{2^{2k+1} (2k+1)!} H_{2k+1}(x) H_{2k+1}(y).
\end{equation}
We first consider the case of $0 < \xi < 1$. We can write
\begin{equation}
\begin{split}
    \eta \frac{\partial S_\xi}{\partial x} &= 2\sum^{+\infty}_{k=0} \frac{(\eta/2)^{2k+1}}{(2k+1)!} H_{2k+1}(x) 
    H_{2k+1}(y) \\
    &= s(\eta) - s(-\eta),
\end{split}
\end{equation}
where $\eta = \sqrt{\xi}$ and $s(\eta)$ is given by the equation \eqref{eq:f}. We 
get the following explicit expression for the partial derivative:
\begin{equation}
    \frac{\partial S_\xi}{\partial x} = \frac{2e^{-\frac{\xi}{1-\xi}(x^2 + y^2)}}{\sqrt{\xi(1-\xi)}} 
    \sinh\left(\frac{2xy\sqrt{\xi}}{1-\xi}\right).
\end{equation}
Integrating, we get
\begin{equation}\label{eq:S2}
\begin{split}
    &S(x, y) - S(0, y) = \frac{\sqrt{\pi}}{2\xi} \left[2 \erf\left(\frac{y}{\sqrt{1-\xi}}\right)\right. \\
    &+\left. \erf\left(\frac{-y+\sqrt{\xi}x}{\sqrt{1-\xi}}\right) - 
    \erf\left(\frac{y+\sqrt{\xi}x}{\sqrt{1-\xi}}\right)\right].
\end{split}
\end{equation}
In contrast with the previous case, $S_\xi(0, y)$ is not zero, so we must compute it separately. 
We have
\begin{equation}\label{eq:S-2}
\begin{split}
    S_\xi(0, y) &= -\sum^{+\infty}_{k=0} \frac{(-\xi)^k}{2^{2k+1}(k+1)!} H_{2k+1}(y) \\
    &= -\frac{1}{2}\sum^{+\infty}_{k=0} \frac{t^k}{(k+1)!} H_{2k+1}(y) \equiv -\frac{1}{2}f(t, y),
\end{split}
\end{equation}
where $t = -\xi/4$. To obtain a compact expression for $f(t, y)$ we use the same approach --- we 
first compute the partial derivative
\begin{equation}
    \frac{\partial}{\partial t} (t f(t, y)) = \sum^{+\infty}_{k=0} \frac{t^k}{k!} H_{2k+1}(y).
\end{equation}
According to Ref.~\cite{*[][{ p. 708, Eq.~(5.12.1-4)}] prudnikov2}, the series on the right-hand 
side is
\begin{equation}
    \sum^{+\infty}_{k=0} \frac{t^k}{k!} H_{2k+1}(y) = \frac{2y}{(1+4t)^{3/2}} \exp\left(\frac{4x^2t}{1+4t}\right).
\end{equation}
We thus have (taking into account that $tf(t,y) = 0$ for $t=0$)
\begin{equation}
    t f(t, y) = -\frac{\sqrt{\pi}}{2} e^{y^2} \left[\erf\left(\frac{y}{\sqrt{1+4t}}\right) - \erf(y)\right].
\end{equation}
Substituting this expression into Eq.~\eqref{eq:S-2}, we get
\begin{equation}
    S_\xi(0, y) = \frac{\sqrt{\pi}}{-\xi} e^{y^2} \left[\erf\left(\frac{y}{\sqrt{1-\xi}}\right) - \erf(y)\right],
\end{equation}
and from Eq.~\eqref{eq:S2} we obtain
\begin{equation}
\begin{split}
    &S_\xi(x, y) = \frac{\sqrt{\pi}}{2\xi} \biggl[2 \erf(y) \\
    &+\erf\left(\frac{-y+\sqrt{\xi}x}{\sqrt{1-\xi}}\right) - \erf\left(\frac{y+\sqrt{\xi}x}{\sqrt{1-\xi}}\right)\biggr].
\end{split}
\end{equation}
We finally arrive to the expression \eqref{eq:bS'-1} for the wave function. The case of $-1 < \xi < 
0$ can be considered analogously.

We also need to show that for $\xi \to 0$ the wave functions given by the expressions 
\eqref{eq:bS'-1} and \eqref{eq:bS'-2} becomes the wave function \eqref{eq:bS'-3}. We have
\begin{equation}\label{eq:d1}
\begin{split}
    \erf\left(\frac{-y+\sqrt{\xi}x}{\sqrt{1-\xi}}\right) &= -\erf(y) + \frac{2}{\sqrt{\pi}} e^{-y^2} \Delta_1 y \\
    &+\frac{1}{2!} \frac{4}{\sqrt{\pi}}ye^{-y^2} (\Delta_1 y)^2 + \mathcal{O}(\xi^{3/2}),
\end{split}
\end{equation}
where
\begin{equation}
    \Delta_1 y = \frac{\sqrt{\xi}x - (1-\sqrt{1-\xi})y}{\sqrt{1-\xi}}.
\end{equation}
The relation \eqref{eq:d1} is valid since
\begin{equation}
    \erf'(z) = \frac{2}{\sqrt{\pi}} e^{-z^2}, \quad \erf''(z) = -\frac{4}{\sqrt{\pi}} z e^{-z^2},
\end{equation}
and $\Delta_1 y = \mathcal{O}(\xi^{1/2})$. Similarly, we can write
\begin{equation}
\begin{split}
    \erf\left(\frac{y+\sqrt{\xi}x}{\sqrt{1-\xi}}\right) &= \erf(y) + \frac{2}{\sqrt{\pi}} e^{-y^2} \Delta_2 y \\
    &-\frac{1}{2!} \frac{4}{\sqrt{\pi}}ye^{-y^2} (\Delta_2 y)^2 + \mathcal{O}(\xi^{3/2}),
\end{split}
\end{equation}
where 
\begin{equation}
    \Delta_2 y = \frac{\sqrt{\xi}x + (1-\sqrt{1-\xi})y}{\sqrt{1-\xi}}.
\end{equation}
Substituting these expressions into Eq.~\eqref{eq:bS'-1} and taking into account that 
\begin{equation}
    \sqrt{-\ln(1-\xi^2)} \sim \xi, \quad \xi \to +0,
\end{equation}
we get that when $\xi \to +0$, Eq.~\eqref{eq:bS'-1} goes to Eq.~\eqref{eq:bS'-3}. The case of $\xi 
\to -0$ can be considered in the same way. This finishes the proof.

\section{Proof of Theorem \ref{thrm:bip}}

The state defined in Eq.~\eqref{eq:psi2} can also be used here, since for this state $\langle\hat{a}\rangle = 
\langle\hat{b}\rangle = 0$, so it remains to be proven only that all bipartite Gaussian states satisfy the inequality 
\eqref{eq:dbS}. Remember that the characteristic function of a bipartite quantum state with density operator 
$\hat{\varrho}$ is defined by $\chi(\vec{t}) = \langle e^{i(\vec{t}, \vec{\hat{X}})} \rangle$, where $\vec{t} = (t_1, 
t_2, t_3, t_4)$ and $\vec{\hat{X}} = (\hat{x}_a, \hat{x}_b, \hat{p}_a, \hat{p}_b)$. A state is called Gaussian if its 
characteristic function is Gaussian, i.e., if it can be written in the following form:
\begin{equation}\label{eq:chiG}
    \chi(\vec{t}) = e^{-(1/2)\vec{t}^{\mathrm{T}} \Sigma \vec{t} + i \vec{m}^{\mathrm{T}} \vec{t}},
\end{equation}
where $\Sigma = (\sigma_{ij})^4_{i,j=1}$ is a real symmetric 4$\times$4 matrix and $\vec{m} = (m_1, m_2, m_3, m_4)$ is a 
real 4-vector. There is no restriction on the vector $\vec{m}$, but to have the characteristic function of a quantum 
state, the matrix $\Sigma$ (the second-order moment covariance matrix) must satisfy the condition 
\cite{PhysRevLett.84.2726},
\begin{equation}\label{eq:sigma2}
    \tilde{\Sigma} = 
    \begin{pmatrix}
        \sigma_{11} & \sigma_{12} & \sigma_{13} + \frac{i}{2} & \sigma_{14} \\
        \sigma_{12} & \sigma_{22} & \sigma_{23} & \sigma_{24} + \frac{i}{2} \\
        \sigma_{13} - \frac{i}{2} & \sigma_{23} & \sigma_{33} & \sigma_{34} \\
        \sigma_{14} & \sigma_{24} - \frac{i}{2} & \sigma_{34} & \sigma_{44}
    \end{pmatrix} 
    \geq 0.
\end{equation}
For any two-mode Gaussian state, this condition is necessary and sufficient for physicality of the state. Due to the 
equality
\begin{equation}
    \langle e^{i t_1 \hat{x}_a} e^{i t_2 \hat{x}_b} e^{i t_3 \hat{p}_a} e^{i 
t_4 \hat{p}_b} \rangle = e^{-\frac{i}{2}(t_1 t_3 + t_2 t_4)} \chi(\vec{t}) 
\equiv \tilde{\chi}(\vec{t}), 
\nonumber
\end{equation}
we can compute the moments $\langle \hat{x}^n_a \hat{x}^m_b \hat{p}^k_a 
\hat{p}^l_b \rangle$ as follows:
\begin{equation}
    \langle \hat{x}^n_a \hat{x}^m_b \hat{p}^k_a \hat{p}^l_b \rangle = (-i)^{n+m+k+l}\left.\frac{\partial^{n+m+k+l} 
    \tilde{\chi}(\vec{t})}{\partial t^n_1 t^m_2 t^k_3 t^l_4}\right|_{\vec{t} = 0}.
\end{equation}
From this we immediately obtain that $\vec{m} = (\langle\hat{r}_j\rangle)^4_{j=1}$ and
\begin{equation}
    \Sigma = \frac{1}{2}(\langle(\Delta\hat{r}_j)(\Delta\hat{r}_k)+(\Delta\hat{r}_k)(\Delta\hat{r}_j)\rangle)^4_{j,k=1},
\end{equation}
where $(\hat{r}_1, \hat{r}_2, \hat{r}_3, \hat{r}_4) = (\hat{x}_a, \hat{x}_b, \hat{p}_a, \hat{p}_b)$. We see that proving 
the inequality \eqref{eq:dbS} for the states with the characteristic function \eqref{eq:chiG} is the same as proving 
the inequality \eqref{eq:bS:1}, $\langle(\hat{a}^{\dagger 2} \pm \hat{b}^2)(\hat{a}^2 \pm \hat{b}^{\dagger 2})\rangle 
\geq 2$, for the states with the characteristic function
\begin{equation}\label{eq:chiG0}
    \chi(\vec{t}) = e^{-(1/2)\vec{t}^{\mathrm{T}} \Sigma \vec{t}}.
\end{equation}
From now on we assume that $\vec{m} = \vec{0}$ and we have to prove the inequality \eqref{eq:bS:1} for all Gaussian 
states with the characteristic function of the form \eqref{eq:chiG0}. In fact, we are going to prove a more strict 
inequality
\begin{equation}\label{eq:ab0}
    \langle\hat{a}^{\dagger 2}\hat{a}^2 + \hat{b}^2\hat{b}^{\dagger 2}\rangle - 2|\langle\hat{a}^2\hat{b}^2\rangle| \geq 
    2
\end{equation}
for all states of the form \eqref{eq:chiG0}. Note that this inequality is invariant with respect to the transformation 
$\hat{a} \to \hat{a}e^{-i \varphi_a}$, $\hat{b} \to \hat{b}e^{-i \varphi_b}$. Since $\hat{a}e^{-i \varphi} = 
\hat{U}^\dagger(\varphi)\hat{a}\hat{U}(\varphi)$, where $\hat{U}(\varphi) = e^{-i \varphi\hat{n}}$ is the phase rotation 
operator, the invariance of the inequality \eqref{eq:ab0} with respect to this transformation means that the left-hand 
side of this inequality is the same for the original state $\hat{\varrho}$ and a transformed state $\hat{\varrho}' = 
(\hat{U}_a(\varphi_a)\otimes\hat{U}_b(\varphi_b))\hat{\varrho} 
(\hat{U}^\dagger_a(\varphi_a)\otimes\hat{U}^\dagger_b(\varphi_b))$ for arbitrary phases $\varphi_a$ and $\varphi_b$. We 
can use the freedom in choosing these phases to simplify the matrix $\Sigma$. The transformed state $\hat{\varrho}'$ is 
of the form \eqref{eq:chiG0} with the matrix $\Sigma'$. For this matrix we have
\begin{equation}
\begin{split}
    \sigma'_{11} = \frac{\sigma_{11} + \sigma_{33}}{2} + \sigma'',  \quad
    \sigma'_{33} = \frac{\sigma_{11} + \sigma_{33}}{2} - \sigma'',
\end{split}
\end{equation}
where $\sigma''$ reads as
\begin{equation}
    \sigma'' = \frac{\sigma_{11}-\sigma_{33}}{2} \cos(2\varphi_a) + \sigma_{13} \sin(2\varphi_a).
\end{equation}
The expressions for $\sigma'_{22}$ and $\sigma'_{44}$ are transformed in a similar way. From this we see that we can 
always choose $\varphi_a$ and $\varphi_b$ such that $\sigma'_{11} = \sigma'_{33}$ and $\sigma'_{22} = \sigma'_{44}$. So, 
we can assume from the beginning that the matrix $\Sigma$ has the property $\sigma_{11} = \sigma_{33} = \sigma_1$ and 
$\sigma_{22} = \sigma_{44} = \sigma_2$. In this case we have
\begin{equation}
\begin{split}
    \langle \hat{a}^{\dagger 2} \hat{a}^2 \rangle &= 2\sigma^2_1 - 2\sigma_1 + \sigma^2_{13} + \frac{1}{2}, \\ 
    \langle \hat{b}^2 \hat{b}^{\dagger 2} \rangle &= 2\sigma^2_2 + 2 \sigma_2 + \sigma^2_{24} + \frac{1}{2}, \\
    |\langle \hat{a}^2 \hat{b}^2 \rangle|^2 &= \frac{1}{4} \biggl[ 
    \Bigl((\sigma_{12}-\sigma_{34})^2+(\sigma_{14}+\sigma_{23})^2\Bigr)^2 \\
    - 4\sigma_{13}&\sigma_{24}((\sigma_{12}-\sigma_{34})^2-(\sigma_{14}+\sigma_{23})^2) + 
    4\sigma^2_{13}\sigma^2_{24}\biggr]. \nonumber
\end{split}
\end{equation}
When we substitute these quantities into the inequality \eqref{eq:ab0} we will get an inequality for the sigmas, which 
must also be derivable from the physicality condition expressed by the inequality \eqref{eq:sigma2} (where, as we have 
assumed, $\sigma_{11} = \sigma_{33} = \sigma_1$ and $\sigma_{22} = \sigma_{44} = \sigma_2$). Let us define
\begin{equation}
    A = 2\sigma^2_1 - 2\sigma_1 + 2\sigma^2_2 + 2\sigma_2 - 1,
\end{equation}
$u = \sigma_{12} - \sigma_{34}$ and $v = \sigma_{14} + \sigma_{23}$, then we have to prove that
\begin{equation}\label{eq:A}
\begin{split}
    (A + \sigma^2_{13} + \sigma^2_{24})^2 &\geq (u^2+v^2)^2 \\
    &- 4 \sigma_{13} \sigma_{24} (u^2-v^2) + 4 \sigma^2_{13} \sigma^2_{24}.
\end{split}
\end{equation}
If we manage to prove that $A \geq u^2 + v^2$, then we will prove the inequality \eqref{eq:A}. In fact, in this case we 
have $A^2 \geq (u^2+v^2)^2$ and  
\begin{equation}
\begin{split}
    2(\sigma^2_{13} + \sigma^2_{24}) A &\geq 4 |\sigma_{13} \sigma_{24}| (u^2+v^2) \\
    &\geq -4 \sigma_{13} \sigma_{24} (u^2-v^2),
\end{split}
\end{equation}
Moreover, $(\sigma^2_{13} + \sigma^2_{24})^2 \geq 4 \sigma^2_{13} \sigma^2_{24}$ independently of $A$ and thus 
Eq.~\eqref{eq:A} follows.

In order to prove that $A \geq u^2 + v^2$ note that for any matrix $P$ the matrix $P^\dagger \tilde{\Sigma} P$ is also 
positive, as well as the matrix $\overline{\Sigma} = \tilde{\Sigma} + P^\dagger \tilde{\Sigma} P$. If we take
\begin{equation}
    P = 
    \begin{pmatrix}
        0 & 0 & -1 & 0 \\
        0 & 0 & 0 & 1 \\
        1 & 0 & 0 & 0 \\
        0 & -1 & 0 & 0
    \end{pmatrix},
\end{equation}
we find that the following matrix is positive:
\begin{equation}
    \overline{\Sigma} = 
    \begin{pmatrix}
        2\sigma_1 & u & i & v \\
        u & 2\sigma_2 & v & i \\
        -i & v & 2\sigma_2 & -u \\
        v & -i & -u & 2\sigma_2
    \end{pmatrix}
    \geq 0.
\end{equation}
From the third-order minor obtained by canceling the fourth row and the fourth column we get the inequality $4\sigma^2_1 
\sigma_2 \geq \sigma_1 (u^2 + v^2) + \sigma_2$. From the third-order minor obtained by canceling the third row and the 
third column we get a similar inequality, $4\sigma_1 \sigma^2_2 \geq \sigma_2 (u^2 + v^2) + \sigma_1$. From these two 
inequalities we obtain
\begin{equation}
    (4\sigma^2_1 - 1) (4\sigma^2_2 - 1) \geq (u^2 + v^2)^2.
\end{equation}
Since $(4 \sigma_1 \sigma_2 - 1)^2 - (4\sigma^2_1 - 1) (4\sigma^2_2 - 1) = 4 (\sigma_1 - \sigma_2)^2 \geq 0$, we have 
$(4 \sigma_1 \sigma_2 - 1)^2 \geq (u^2 + v^2)^2$. Due to the inequalities $\sigma_1 \geq 1/2$ and $\sigma_2 \geq 1/2$ we 
obtain
\begin{equation}\label{eq:uv}
    4 \sigma_1 \sigma_2 - 1 \geq u^2 + v^2.
\end{equation}
From the non-negativity of the determinant $\det\tilde{M}$ we get
\begin{equation}\label{eq:uv2}
    (u^2 + v^2 - 4 \sigma_1 \sigma_2 + 1)^2 - 4 (\sigma_1 - \sigma_2)^2 \geq 0.
\end{equation}
Due to the inequality \eqref{eq:uv}, the inequality \eqref{eq:uv2} is equivalent to $4 \sigma_1 \sigma_2 - 1 - u^2 - v^2 
\geq 2 |\sigma_1 - \sigma_2|$, so we have $4\sigma_1 \sigma_2 - 2|\sigma_1 - \sigma_2| - 1 \geq u^2 + v^2$. If we prove 
that $A \geq 4\sigma_1 \sigma_2 - 2|\sigma_1 - \sigma_2| - 1$, we are done. After equivalent transformations this 
inequality becomes
\begin{equation}
    (\sigma_1 - \sigma_2)^2 + |\sigma_1 - \sigma_2| - (\sigma_1 - \sigma_2) \geq 0,
\end{equation}
which is obviously valid.

\end{document}